\documentclass[12pt]{amsart}

\usepackage{amsthm,amssymb,amsmath}
\usepackage{hyperref}

 \usepackage[english]{babel}
\usepackage[dvips]{graphicx}
\usepackage{latexsym,amsfonts,amssymb}
\usepackage{amsfonts}

\usepackage{amscd}
\usepackage[all]{xy}
\usepackage{textcomp} 
\usepackage{mathrsfs} 

\def\be{\begin{equation}}
\def\ee{\end{equation}}

\def\>>{\rangle}

\def\fr12{\frac{1}{2}}

\def\bx{{\bf x}}
\def\bz{{\bf z}}
\def\bw{{\bf w}}
\def\bs{{\bf s}}

\def\cC{\mathcal C}
\def\cQ{\mathcal Q}
\def\cH{\mathcal H}
\def\cM{\mathcal M}

\def\cS{\mathcal S}
\def\cW{\mathcal W}
\def\cY{\mathcal Y}

\def\bC{\mathbb C}

\def\bZ{\mathbb Z}

\def\Rc{\check{R}}

\def\pp{\phantom{.}}

\def\Y1{Y^{(1)}}

\DeclareMathOperator{\End}{End}

\newtheorem{theorem}{\emph {Theorem}}

\newtheorem{proposition}{\emph {Proposition}}

\begin{document}	

\title[ASEP with open boundaries and Koornwinder
polynomials]{Asymmetric Simple Exclusion Process with open boundaries
  and  Koornwinder polynomials} 
\author[L. Cantini]{Luigi Cantini}
\address{LPTM, Universit\'e de Cergy-Pontoise (CNRS UMR
    8089), 
Cergy-Pontoise Cedex, France.}
\date{\today}
\email{luigi.cantini@u-cergy.fr}
\thanks{The work of LC is partially supported by CNRS through a
  ``Chaire d'excellence''. It is a pleasure to thank Jan de Gier for
  collaboration at an erly stage of this project and the Departement
  of Mathematics and Statistics of the University of
  Melbourne for kind hospitality.} 

\begin{abstract}

In this paper we analyze the steady state of the Asymmetric Simple
Exclusion process with open boundaries and second class particles by
deforming it through the introduction of spectral parameters. The
(unnormalized) probabilities of the particle configurations get
promoted to Laurent polynomials in the spectral parameters and
are constructed in terms of non-symmetric Koornwinder polynomials. In
particular we show that the partition function coincides with a
symmetric Macdonald-Koornwinder polynomial. As an outcome we compute
the steady current and the average density of first class particles.

\end{abstract}

\begin{titlepage}

\maketitle

\end{titlepage}




\section{Introduction}

The Asymmetric Simple Exclusion Process (ASEP) is a paradigmatic example of
an out of equilibrium systems
\cite{spitzer1970interaction,derrida1998exactly,chou2011non,derrida2007non}. In
its simplest form the ASEP consists of particles located on the sites
of a directed one dimensional lattice under the condition that each
site can be occupied by at most one particle, so that the local
configurations on each site can be denoted by a $\bullet$ for an
occupied site and by a $\circ$ for an empty site. The particles are
subject to a stochastic evolution which consists of jumps of 
a particle on an empty neighboring sites with
rates $t^\fr12$ or $t^{-\fr12}$ depending on the direction of the jump. 
\begin{align*}
\bullet \circ \xrightarrow{t^{\fr12}} \circ \bullet &&
\circ \bullet  \xrightarrow{t^{-\fr12}} \bullet \circ 
\end{align*}
Despite its apparent simplicity, this model can be successfully employed to
describe or at least to capture the main features of very different
physical systems. Indeed the ASEP has appeared 
for the first time in biology \cite{macdonald1968kinetics} but since
then it has found applications in the 
study of a wide range of physical phenomena: traffic flow, surface
growth, sequence alignment, etc. (see \cite{chou2011non} for a recent review of many of these applications). 

The ASEP can be in particular used to model
the exchange of particles
between two reservoirs at different chemical potential \cite{derrida2007non}.
In this case one considers a lattice of finite length $N$ and the
particles can be exchanged with a left and a right reservoir \cite{sandow1994partially}. So,
besides a rate for right and left jumps of a particle, we have also
rates $\alpha, \gamma$, respectively for a particle injected in or
removed from the leftmost site and $\delta,\beta$ for a particle
injected in or removed from the rightmost site.  
\begin{center}
\vskip .7cm
\includegraphics[width= 4truecm]{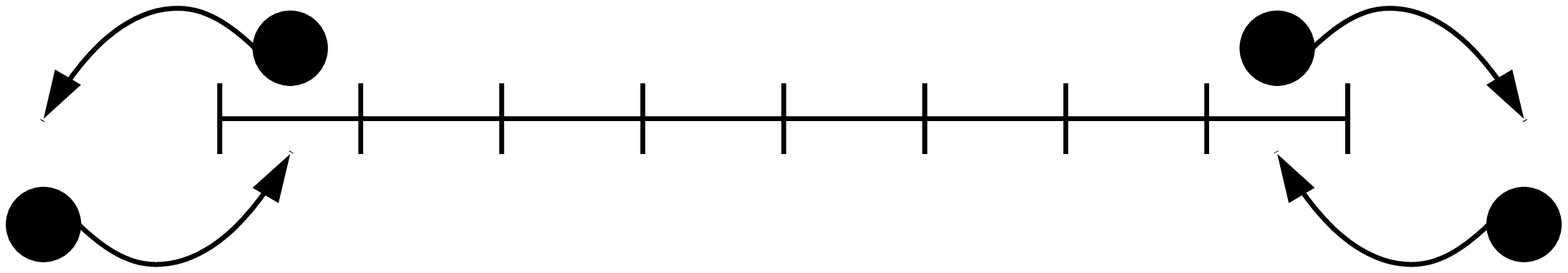}
\vskip .6cm
\begin{picture}(0,0)
\put(-50,20){$\alpha$}
\put(45,20){$\delta$}
\put(-50,60){$\gamma$}
\put(45,60){$\beta$}
\end{picture}
\end{center}
\vskip -.3cm
Apart for its physical interest, this model has recently drawn the
attention also of combinatorialists. Indeed the (unnormalized) stationary
probability of the particle configurations are given by enumerations
of certain classes of tableaux  
\cite{duchi2005combinatorial,corteel2007tableaux,corteel2011tableaux}.
  
The model we are going to consider here has as an extra feature, the
presence of \emph{second class particles}, which are confined in the
finite lattice, i.e. that are not 
exchanged with the reservoirs at the boundaries. 
These second class particles can be thought as some sort of 
mobile impurities and will be indicated by $\ast$. Their jump rates on
empty sites, or exchange rates with usual (\emph{first class})
particles are given by
\begin{align*}
\ast \circ \xrightarrow{t^{\fr12}} \circ \ast && \bullet \ast
\xrightarrow{t^{\fr12}} \ast \bullet &\\
\circ \ast \xrightarrow{t^{-\fr12}} \ast \circ && \ast\bullet
\xrightarrow{t^{-\fr12}} \bullet\ast & 
\end{align*}
We shall call this model for short ``open 2ASEP''. 

The Markov chain of the 2ASEP turns out to be integrable for a generic
choice of the parameters introduced above. Actually it is not
difficult to show that the ASEP (in absence of second class particles)
is algebraically equivalent to a spin 
$\frac{1}{2}$ chain of XXZ type with non diagonal boundaries. 
A great deal of effort has been devoted to exploit the integrable
structure of the ASEP, either through Bethe ansatz techniques or
separation  of variables or other powerfull methods of the theory of integrable
systems, in order to study the full spectrum of the ASEP Markov
matrix. Reviewing all these studies goes beyond  the scopes of the present
paper, the interested reader can have a look at the nice introduction
of \cite{lazarescu2014bethe} and references therein.

A different approach, which goes under the name of Matrix Product
Ansatz \cite{derrida1993exact}, has proven to be more successful in computing
the properties of the asymptotic steady state of the system 
\cite{blythe2007nonequilibrium}. This approach consists in writing the
stationary probability of any particles configuration as a trace of
products of some auxiliary (infinite dimensional) matrices that satisfy a certain
algebra. For concrete computations one has to find manageable
representations of the matrix algebra. 

In absence of second class particles, Uchiyama, Sasamoto and Wadaty
\cite{uchiyama2004asymmetric} have  provided a representation of the
Matrix algebra in terms of the 
Askey-Wilson polynomials, which allows for the exact computation at
finite size of the steady state current and of its the density. These
quantities are given by rather simple expression involving the moments
of the Askey-Wilson measure. Combining these results of  
\cite{uchiyama2004asymmetric} with a remarkable combinatorial
expression for the stationary measure of particle configurations in
terms of staircase tableaux \cite{corteel2011tableaux}, Corteel and Williams
have provided a very interesting
combinatorial formula for the moments of the Askey-Wilson polynomials \cite{corteel2011tableaux}
(see also \cite{corteel2012formulae}). 

In \cite{uchiyama2008two}, Uchiyama has tackled the case in presence of
second class particles, using a representation of the
matrix algebra in terms of a t-boson. He was able to compute the so
called ``partition function'' as a certain contour integral and from it
to compute the 
phase diagram of the model in the thermodynamic limit. Nonetheless his
finite size formulas are quite complicated if compared with the one
obtained in absence of second class particles in
\cite{uchiyama2004asymmetric} using the Askey-Wilson polynomials.

In the present paper we put forward an approach to
the computation of the steady state of the open
2ASEP that bypasses the Matrix Product Ansatz. Such an approach has
been pioneered by Di~Francesco and Zinn-Justin in the context of the
stochastic dense $O(1)$ loop model  \cite{di2005around} and  
is based just on the integrability of the model.
In this introduction we are going
to present just the rough idea behind this approach, whose details are
discussed in the body of the paper.

Given our choice of jump rates we know that the open 2ASEP is integrable, which
means in particular that its Markov matrix $\cM$ can be obtained
starting from the so called \emph{scattering matrices} $\cS_i$. These
matrices will be defined in Section 
\ref{integrability}, here we only remark that they depend on a set of
spectral parameters $\bz= \{z_1,\dots,z_N\}$, where $N$ is length of
the lattice and the the parameter $z_i$ is ``attached'' the $i$-th
site. The vector $\Psi_{N,m}$,  gathering the stationary probabilities
of each 2ASEP configurations on a system of length $N$ in the sector
with $m$ second class particles, can be obtained by the eigenvector
$\Psi_{N,m}(\bz)$ of the scattering matrices
$$
\cS_i(\bz)\Psi_{N,m}(\bz)= \Psi_{N,m}(\bz)
$$
upon specialization $z_i=1$
$$
\Psi_{N,m}= \Psi_{N,m}({\bf 1}).
$$
The advantage of dealing with $\Psi_{N,m}(\bz)$ is that, this vector can be
normalized in such a way that it is a Laurent polynomial in the
variables $\bz$ and it satisfies a set of exchange-reflection 
relations that correspond to a degeneration of the quantum Knizhnik-Zamolodchikov (qKZ)
equations of type $\tilde \cC_N$ (see
\cite{kasatani2009boundary,shigechi2014laurent} for some recent
results on  
polynomial solutions of qKZ equations of type $\tilde \cC_N$). We
shall construct the solution of the qKZ equations using non-symmetric
Koornwinder polynomials and study some of their properties (notice the
recent work \cite{stokman2014koornwinder} on the relations between
Koornwinder polynomials and the XXZ spin chain). 
We shall prove that the normalization of $\Psi_{N,m}(\bz)$
(sometimes called the ``partition function'') is
given by a symmetric Macdonald-Koornwinder polynomial for a single
column partition of length $N-m$. 
From the weighted partition function we shall extract current and
average density of first class particles. Finally, 
using a known integral representation for the  Macdonald-Koornwinder
polynomials we shall compute these quantities in the thermodynamic
limit, in the case of a fixed density of second class particles,
recovering the phase diagram of \cite{uchiyama2008two}.  

\vskip 1cm

The paper is organized as follows. In Section \ref{integrability} we
shall discuss the integrability of the open 2ASEP. In particular we
explain how to make to model inhomogeneous 
by introducing the spectral parameters $\bz$ and show in Theorem
\ref{scatt-to-exch} that the inhomogeneous generalization of the
stationary probability $\Psi(\bz)$ satisfies a set of
exchange-reflection relations. 
The algebraic structures behind these equations will
be discussed in Appendix \ref{algebraic}, where we also recall some
basic facts about symmetric Macdonald-Koornwinder polynomials,
non-symmetric Koornwinder polynomials and Askey-Wilson polynomials.
In Section \ref{analysis} we analyze the solution of the qKZ
equations. We show in Theorems \ref{theo-mac-koorn}, \ref{gener-funct} that the
normalization of  $\Psi(\bz)$ and the generating function for the
total density of first class particles are given by symmetric
Macdonald-Koornwinder polynomials. In Section \ref{sect-rec-rel} we
shall discuss some recursion relations between solutions of the
exchange-reflection equations for systems of different size and different values
of the parameters $a,b,c,d$.  
In Section \ref{sect-phase} we compute the average current and density
of first class particles, we shall then take the thermodynamic limit and compute the
phase diagram for the model.

\section{Integrability and exchange-reflection equations}\label{integrability}

The Markov matrix which generates the time evolution of the 2ASEP with
second class particles can be written as
a sum of local terms. Let us associate to each site a vector space
$V\simeq \bC^3$ with basis vectors $\{v_{-1},v_{0}, v_{1}\}$ which
correspond to the following local particle configurations
\begin{align*}
v_{-1}&\longleftrightarrow \circ \textrm{~~~~~empty site}\\
v_{0}&\longleftrightarrow \ast \textrm{~~~~~second class particle}\\
v_{1}&\longleftrightarrow \bullet \textrm{~~~~~first class particle}
\end{align*}
The vector space with a basis labeled by global particle configurations has a
tensor product structure $H_N=V_1\otimes V_2\otimes \cdots \otimes
V_N$, with $V_i \simeq V \simeq \bC^3$.
The local contributions to the Markov chain are given by the
following operators ${\bf e}\in \End 
(V\otimes V)$, 
${\bf f}(\rho,\sigma) \in \End(V)$
\be
\begin{split}
&{\bf e}= \sum_{i,j=-1}^1\left(E^{(i,j)}\otimes
  E^{(j,i)}-E^{(j,j)}\otimes
  E^{(i,i)}\right)t^{\frac{1}{2}\textrm{sign}(j-i)}\\ 
&{\bf f}(\rho,\sigma) = \rho \left(E^{(-1,1)}-E^{(1,1)}\right) +\sigma
\left(E^{(1,-1)}-E^{(-1,-1)}\right)  
\end{split}
\ee
where the operators $E^{(j,i)}\in \End(V)$ are defined by
$E^{(j,i)}v_k=\delta_{j,k}v_i$. 
The Markov matrix $\cM$ reads 
\be 
\cM= \sum_{i=1}^{N-1} {\bf e}_i + {\bf f}_1(\gamma,\alpha) +{\bf f}_N(\beta,\delta).
\ee
where 
$$
{\bf e}_i= {\bf 1}_1\otimes \cdots \otimes {\bf 1}_{i-1}\otimes {\bf e}
\otimes {\bf 1}_{i+2} \otimes \cdots \otimes {\bf 1}_{N}
$$
$$
{\bf f}_i(\rho,\sigma)= {\bf 1}_1\otimes \cdots \otimes {\bf 1}_{i-1}\otimes {\bf f}(\rho,\sigma)
\otimes {\bf 1}_{i+1} \otimes \cdots \otimes {\bf 1}_{N}
$$
Since the dynamics preserves the number of second class particles, we have a
splitting of $H_N$ as a direct sum 
$$
H_N=\bigoplus_{m=0}^N H_{N,m}
$$
of subspaces (sectors) $H_{N,m}$ with fixed number $m$ of second class
particles, which are preserved by the action of $\cM$.

\vskip .5cm

A first glimpse of the algebraic structure behind the Markov matrix
$\cM$ comes from the observation that the operators ${\bf e}_i, {\bf
  f}_1(\gamma,\alpha)$ and ${\bf f}_N(\beta,\delta)$ generate a
representation of  $\cH_N$, the affine Hecke algebra of type
$\hat \cC_N$. Indeed, the operators $T_0,T_1,\dots,T_N$ defined as 
$$
T_0 =  \alpha^{-\fr12}\gamma^{-\fr12}{\bf f}_1(\gamma,\alpha) +
\alpha^{\fr12}\gamma^{-\fr12}  {\bf 1}
$$ 
$$
T_N =  \beta^{-\fr12}\delta^{-\fr12}{\bf f}_N(\beta,\delta) +
\beta^{\fr12}\delta^{-\fr12}  {\bf 1}
$$ 
and for $1 \leq i \leq N-1$
$$
T_i = {\bf e}_i+ t^{-\fr12}{\bf 1}
$$
satisfy the commutation relations of the generators of $\cH_N$ (see
Section \ref{ahaCN})
\begin{align}
T_i-T_i^{-1}&= t_i^{\fr12}-t_i^{-\fr12}\\
T_iT_j&=T_jT_i ~~\textrm{if}~~|i-j|>1\\ 
T_i T_{i+1} T_i &= T_{i+1} T_i T_{i+1}~~\textrm{if}~~i\neq 0,N-1\\
T_0 T_1 T_0 T_1 &= T_1 T_0 T_1 T_0\\
T_N T_{N-1} T_N T_{N-1} &= T_{N-1} T_N T_{N-1}T_N
\end{align}
with $t_0^{\fr12}= \alpha^{\fr12}\gamma^{-\fr12}, t_N^{\fr12}=
\beta^{\fr12}\delta^{-\fr12}$ and $t_i=t$ for $1\leq i\leq N-1$.

\vskip .5cm

Operators written as sums of generators of $\cH_N$, like $\cM$ are
known to be quantum integrable
\cite{sklyanin1988boundary,doikou2010introduction}. 
In what follows we shall present only a few notions of the theory of
quantum integrable systems needed in order to study the
properties of the stationary measure.

The first ingredient are the $\Rc$ matrices based on the baxterization of the Hecke algebra of
type $A_n$, which read
\be
\label{Ri}
\Rc_{i}(z)={\bf 1} +
\frac{z-1}{t^{\fr12}z-t^{-\fr12}} {\bf e}_i.  
\ee
They satisfy the braided Yang-Baxter equation (YBE)
$$
\Rc_{i}(yz^{-1})\Rc_{i+1}(xz^{-1})\Rc_{i}(xy^{-1})=
\Rc_{i+1}(xy^{-1})\Rc_{i}(xz^{-1})\Rc_{i+1}(yz^{-1}) 
$$
and the so called unitarity condition
$$
\Rc_{i}(z)\Rc_{i}(z^{-1})= {\bf 1}.
$$
The second ingredient are the boundary scattering matrices. A
classification of all the integrable boundary scattering matrices
compatible with the matrices $\Rc_{i}(z)$ has been obtained
recently \cite{crampe2014open}, here we shall consider the case that
corresponds to the baxterization of
the boundary matrices ${\bf f}_1(\delta,\alpha)$ and ${\bf
  f}_N(\beta,\gamma)$. For $t\neq 1$ it turns out to be convenient to
parametrize the boundary rates as 
\be\label{parametr}
\begin{split}
\alpha= \frac{(t^{\fr12}-t^{-\fr12})ab}{(a-1)(b-1)},~~~~~\gamma=
\frac{t^{-\fr12}-t^{\fr12}}{(a-1)(b-1)} \\
\beta= \frac{(t^{\fr12}-t^{-\fr12})cd}{(c-1)(d-1)},~~~~~\delta=
\frac{t^{-\fr12}-t{\fr12}}{(c-1)(d-1)} 
\end{split}
\ee
This means that for $t>1$ we can choose $a,c<0$ and $b,d>1$, while for
$0<t<1$ we have $a,c<0$ and $0<b,d<1$.  
The boundary scattering matrices read
\be\label{KL}
K_1(z|a,b)= {\bf 1} +
\frac{(z^2-1)}{(z- a)(z-
  b) }\gamma^{-1}{\bf f}_1(\gamma,\alpha), 
\ee
\be\label{KR}
K_N(z|c,d)= {\bf 1} +
\frac{(1-z^2)}{(c
  z-1)(d z-1) }\delta^{-1}{\bf f}_N(\beta,\delta).
\ee
They satisfy the boundary Yang-Baxter equations (BYBE), also called reflection
equations \cite{sklyanin1988boundary,cherednik1984factorizing}
\begin{multline}
\Rc_{1}(xy^{-1})K_1(y)\Rc_{1}(x^{-1}y^{-1}))K_1(x) =\\
K_1(x)\Rc_{1}(x^{-1}y^{-1})K_1(y)\Rc_{1}(xy^{-1}), 
\end{multline}
\begin{multline}
\Rc_{N-1}(xy^{-1})K_N(x)\Rc_{N-1}(xy)K_N(y) =\\
K_N(y)\Rc_{N-1}(xy)K_N(x)\Rc_{N-1}(xy^{-1}),
\end{multline}
and the unitarity conditions
$$
K_1(x)K_1(x^{-1})= {\bf 1},~~~~~K_N(x)K_N(x^{-1})= {\bf 1}.
$$
Here is the explicit form of the $K$ matrices acting respectively on
the first and last site
\begin{equation}\label{KLexpl}
K_1(z|a,b)= 
\left(\begin{array}{ccc}
1 & 0 & 0\\ 
0 & 1 & 0\\ 
0 & 0 & 1
\end{array}\right) +
\frac{z^2-1}{(z- a)(z-
  b) }\left(
\begin{array}{ccc}
-1 & 0 & - ab\\
0 & 0 & 0\\ 
1 & 0 & ab
\end{array}
\right), 
\end{equation}
\begin{equation}\label{KRexpl}
K_R(z|c,d)= 
\left(\begin{array}{ccc}
1 & 0 & 0\\ 
0 & 1 & 0\\ 
0 & 0 & 1
\end{array}\right)+
\frac{1-z^2}{ (cz-1)(dz-1) }
\left(
\begin{array}{ccc}
cd  & 0 & 1\\
0 & 0 & 0\\
-cd & 0 & -1\\
\end{array}
\right).
\end{equation}

The most common way to exploit the algebraic properties of the $\Rc$
and $K$ matrices is to use such matrices as building blocks of the so
called double row transfer matrix, which depends on a 
spectral parameter $y$ and which commute with the matrix $\cM$ for any
values of $y$. The diagonalization of $\cM$ is then turned into the
diagonalization of the transfer matrix \cite{sklyanin1988boundary}.

Here instead we take a slightly different approach, which will work
effectively for the the determination 
of the stationary state of $\cM$, and is analogous to the one already
employed in \cite{cantini2009qkz}, for the stochastic dense $O(1)$
loop model with open boundaries. Instead of using double row
the transfer matrix, we use the so called Scattering Matrices
which are defined by
\begin{multline}
\cS_i(\bz) = \Rc_{i-1}(z_i z_{i-1}^{-1})\cdots\Rc_{2}(z_i z_{2}^{-1})\Rc_{1}(z_i z_{1}^{-1})\\
K_1(z_i^{-1})\Rc_{1}(z_i z_{1})\cdots \Rc_{i-2}(z_i z_{i-2})\Rc_{i-1}(z_i z_{i-1})\\
\Rc_{i}(z_i z_{i+1})\cdots \Rc_{N-2}(z_i z_{N-1})\Rc_{N-1}(z_i
  z_{N})  \\
K_N(z_i)\Rc_{N-1}(z_i/z_{N})\cdots\Rc_{i+1}(z_i/z_{i+2})\Rc_i(z_i/z_{i+1}), 
\end{multline}
where $\bz =\{z_1,z_2,\dots,z_N\}$.
Even though the matrices  $\cS_i(\bz)$ do not commute with $\cM$, they
commute among themselves 
$$
[\cS_i(\bz),\cS_j(\bz)]=0,
$$
and in the limit $\bz \rightarrow {\bf 1}$ 
their eigenvectors coincide with the eigenvectors of $\cM$.
Indeed, let $\Psi(\bz)$ be a common eigenvector of $\cS_i(\bz)$, with
eigenvalues $\Lambda_i(\bz)$
\be\label{scatt-0}
\cS_i(\bz)\Psi(\bz)=\Lambda_i(\bz) \Psi(\bz).
\ee
Call $\cS_i(z):=\cS_i(z_{j\neq i}=1,z_i=z), \Lambda(z)=\Lambda(z_{j\neq i}=1,z_i=z)$ and
$\Psi_i(z):= \Psi(z_{j\neq i}=1,z_i=z)$, so that the specialization $z_{j\neq
  i}=1$ and $z_i=z$ of eqs.(\ref{scatt-0}) reads
$$
\cS_i(z)\Psi_i(z)=\Lambda_i(z)\Psi_i(z). 
$$
Differentiating this equation with respect to $z$ and then setting $z=1$
one gets 
$$
\cS_i'(1)\Psi_i(1)+\cS_i(1)\Psi_i'(1)=\Lambda_i(1)\Psi_i'(1)+\Lambda'_i(1)\Psi_i(1)
$$
but $\Psi_i(1)=\Psi({\bf 1}), \cS_i(1)={\bf 1}, \Lambda_i(1)=1$ and it s simple to check that 
\be
\cS_i'(1) = \frac{2}{t^{\fr12}-t^{-\fr12}}\cM.
\ee 
In this way we have turned the problem of the diagonalization of $\cM$
into the
simultaneous diagonalization of the Scattering Matrices
$\cS_i(\bz)$. 

The common eigenvectors of $\cS_i(\bz)$ have nice covariance
properties under the action of the $\Rc$ and $K$ matrices.
Let $\Psi(\bz)$ be a common eigenvector of $\cS_i(\bz)$  
with eigenvalues $\Lambda_i(\bz)$ and define
\begin{align*}
\Psi_i(\bz) &:=\Rc_i(z_iz_{i+1}^{-1})\Psi(\bz),\\
\tilde \Psi_1(\bz)&:=K_1(z_1)\Psi(\bz),\\
\tilde \Psi_N(\bz)&:=K_N(z_N)\Psi(\bz).
\end{align*}
The vectors $\Psi_i(\bz), \tilde \Psi_1(\bz)$ and $\tilde \Psi_N(\bz)$
are common eigenvectors respectively of 
$$
{\bf s}_i\cS_j(\bz) , \textrm{~~~~~~~~}{\bf s}_0\cS_i(\bz)
\textrm{~~~~and~~~~} 
{\bf s}_N\cS_i(\bz),
$$
where ${\bf s}_i$  for $i\neq 0,N$ acts on functions of $\bz$ by
exchanging $z_i\leftrightarrow z_{i+1}$, while ${\bf s}_0$, ${\bf s}_N$
exchange respectively $z_1\leftrightarrow z_1^{-1}$, $z_N\leftrightarrow
z_N^{-1}$. 

In particular
\be\label{eigen-R}
\begin{split}
\left({\bf s}_i\cS_j(\bz)\right) \Psi_i(\bz) &= \Lambda_j(\bz)
\Psi_i(\bz) \textrm{~~~~~for~~~~~} j\neq
i,i+1\\
\left({\bf s}_i\cS_i(\bz)\right) \Psi_i(\bz) &= \Lambda_{i+1}(\bz)
\Psi_i(\bz)\\
\left({\bf s}_i\cS_{i+1}(\bz)\right) \Psi_i(\bz) &= \Lambda_{i}(\bz)
\Psi_i(\bz)
\end{split}
\ee
\be\label{eigen-K1}
\begin{split}
\left({\bf s}_0\cS_j(\bz)\right) \tilde \Psi_1(\bz) &= \Lambda_j(\bz)
\tilde \Psi_i(\bz) \textrm{~~~~~for~~~~~} j\neq
1\\
\left({\bf s}_0\cS_1(\bz)\right) \tilde \Psi_1(\bz) &= \Lambda_1^{-1}(\bz)
\tilde \Psi_i(\bz)
\end{split}
\ee
\be\label{eigen-KN}
\begin{split}
\left({\bf s}_N\cS_j(\bz)\right) \tilde \Psi_N(\bz) &= \Lambda_j(\bz)
\tilde \Psi_N(\bz) \textrm{~~~~~for~~~~~} j\neq
N\\
\left({\bf s}_N\cS_N(\bz)\right) \tilde \Psi_N(\bz) &= \Lambda_N^{-1}(\bz)
\tilde \Psi_N(\bz)
\end{split}
\ee
Eqs.(\ref{eigen-R}-\ref{eigen-KN}) are consequences of the following
commutation relations that are immediate to verify using the YBE, BYBE
and unitarity equations
\begin{equation}
\begin{split}
\Rc_i(z_iz_{i+1}^{-1})\cS_j(\bz) &= \left({\bf s}_i
\cS_j(\bz)\right)\Rc_i(z_iz_{i+1}^{-1})\textrm{~~~~~for~~~~~} j\neq
i,i+1\\
\Rc_i(z_iz_{i+1}^{-1})\cS_i(\bz)&=\left({\bf s}_i
\cS_{i+1}(\bz)\right)\Rc_i(z_iz_{i+1}^{-1})\\
\Rc_i(z_iz_{i+1}^{-1})\cS_{i+1}(\bz)&=\left({\bf s}_i
\cS_{i}(\bz)\right)\Rc_i(z_iz_{i+1}^{-1})
\end{split}
\ee
\be
\begin{split}
K_1(z_1)\cS_j(\bz)&= \left({\bf s}_0
\cS_j(\bz)\right)\textrm{~~~~~for~~~~~} j\neq
1\\
K_1(z_1)\cS_1(\bz)&= \left({\bf s}_0
\cS^{-1}_1(\bz)\right)
\end{split}
\ee
\be
\begin{split}
K_N(z_N)\cS_j(\bz)&= \left({\bf s}_N
\cS_j(\bz)\right)\textrm{~~~~~for~~~~~} j\neq
N\\
K_N(z_N)\cS_N(\bz)&= \left({\bf s}_0
\cS^{-1}_N(\bz)\right)
\end{split}
\ee

Since the $\Rc$ and $K$ matrices are stochastic and preserve the
subspace $H_{N,m}$, the stationary state in this sector
is lifted to $\Psi_{N,m}(\bz)$, the \emph{unique} solution of the
following equation on $H_{N,m}$
\be\label{scatt-1}
\cS_i(\bz) \Psi_{N,m}(\bz)=\Psi_{N,m}(\bz).
\ee
Though such equations are still to difficult to deal with, they tell
us that 
$\Psi_{N,m}(\bz)$ can be normalized to be polynomial in the spectral
parameters $\bz$. To really get a control on the solution of
eqs.(\ref{scatt-1}) we need the following 
\begin{theorem}\label{scatt-to-exch}
The solution $\Psi_{N,m}(\bz)$  of eqs.(\ref{scatt-1}) can be normalized in
such a way that it satisfies the following set of exchange-reflection equations 
\begin{align}\label{exch-open1}
\Rc_i(z_iz^{-1}_{i+1})\Psi_{N,m}(\bz) &= {\bf s}_i  
\Psi_{N,m}(\bz)\\ \label{exch-open2}
K_1(z_1)\Psi_{N,m}(\bz) &= {\bf s}_0 
\Psi_{N,m}(\bz)\\ \label{exch-open3}
K_N(z_N) \Psi_{N,m}(\bz) &= {\bf s}_N  
\Psi_{N,m}(\bz).
\end{align}
\end{theorem}
\begin{proof}
As we argued above, $\Psi_{N,m}(\bz)$ can be
normalized to be a polynomial in the spectral parameters $\bz$ with no
overall factor. Call such minimal polynomial solution $\bar
\Psi_{N,m}(\bz)$. Using eqs.(\ref{eigen-R}-\ref{eigen-KN}) and the
uniqueness of the solution of eqs.(\ref{scatt-1}) we conclude that $\bar
\Psi_{N,m}(\bz)$ satisfies eqs.(\ref{exch-open1}-\ref{exch-open3}), up
to a multiplicative factor  
\begin{equation}\label{exch-interm}
\begin{split}
\Rc_i(z_iz^{-1}_{i+1})\bar \Psi_{N,m}(\bz) &= f_i(\bz)~{\bf s}_i  
\bar\Psi_{N,m}(\bz)\\
K_1(z_1)\bar\Psi_{N,m}(\bz) &= f_0(\bz)~{\bf s}_0 
\bar\Psi_{N,m}(\bz)\\
K_N(z_N) \bar\Psi_{N,m}(\bz) &= f_N(\bz)~{\bf s}_N  
\bar\Psi_{N,m}(\bz).
\end{split}
\ee
Using the analytic structure of the right hand side of
eqs.(\ref{exch-interm}) and the unitarity conditions, we conclude that the
functions  $f_i(\bz)$ can take the form 
\begin{align*}
f_i(\bz) &= (-1)^{\kappa_i}
\left(\frac{t^{\fr12}z_{i+1}-t^{-\fr12}z_i}{t^{\fr12}z_i-t^{-\fr12}z_{i+1}}
  \right)^{\epsilon_i}  ~~~\textrm{for}~~~i\neq 0,N\\
f_0(\bz) &= (-1)^{\kappa_0} z_1^{d_1}
\left(\frac{(1-az)(1-bz)}{(z-a)(z-b)}
  \right)^{\epsilon_0}\\
f_N(\bz) &= (-1)^{\kappa_n} z_N^{d_N}
\left(\frac{(z-c)(z-d)}{(1-cz)(1-dz)}
  \right)^{\epsilon_i}
\end{align*}
for some $\kappa_i,\epsilon_i\in\{0,1\}$ and  
where $d_i$ is the degree of $\bar\Psi_{N,m}(\bz)$ in $z_i$.
Setting $z_i=z_{i+1}$ in the first of the eqs.(\ref{exch-interm}) leads
to  
$$
\bar \Psi_{N,m}(\bz)_{z_i=z_{i+1}}= (-1)^{\kappa_i}\bar \Psi_{N,m}(\bz)_{z_i=z_{i+1}}
$$
which forces the $\kappa_i=0$, otherwise one would find that
$\bar\Psi_{N,m}(\bz)$ is divisible by $z_i-z_{i+1}$. In the same way
setting $z_1 = \pm 1$ in the second two or $z_N=\pm 1$ in third of
eqs.(\ref{exch-interm}) 
one get $\kappa_0=\kappa_N=0$ and moreover that $d_1$ and $d_N$ must be even integers. 
 
The last step of the proof consist in combining eqs.(\ref{scatt-1})
and eqs.(\ref{exch-interm}). This leads to $\epsilon_i=0$ and
$d_1=d_N=d$, therefore it is immediate to conclude that the Laurent
polynomial 
$$
\Psi_{N,m}(\bz) := \prod_{i=1}^N z_i^{-d/2} ~\bar \Psi_{N,m}(\bz)_{z_i=z_{i+1}}
$$
satisfies eqs.(\ref{exch-open1}-\ref{exch-open3}).
\end{proof}
A simple consequence of Theorem \ref{scatt-to-exch} is the following
uniqueness result
\begin{theorem}\label{unicity}
The system of equations (\ref{exch-open1}-\ref{exch-open3})  admits a 
Laurent polynomial solution, unique up to a multiplication of a
function invariant under the action of ${\bf s}_i$.
\end{theorem}

\section{Analysis of the solution of the exchange-reflection equations}\label{analysis}

In order to analyze the solution of the exchange-reflection equations 
(\ref{exch-open1}-\ref{exch-open3}) we need to introduce some algebraic
background. 

\subsection{Affine Hecke algebra of type $\tilde \cC_N$ and Noumi representation}\label{ahaCN}

The affine Weyl group $\cW_N$ of type $\tilde \cC_N$ is the group
generated by elements $\bs_0,\bs_1,\dots,\bs_N$ subject to the
relations 
\begin{align}
\bs_i^2&=1\\
\bs_i\bs_j&=\bs_j\bs_i ~~\textrm{if}~~|i-j|>1\\ 
\bs_i \bs_{i+1} \bs_i &= \bs_{i+1} \bs_i \bs_{i+1}~~\textrm{if}~~i\neq 0,N-1\\
\bs_0 \bs_1 \bs_0 \bs_1 &= \bs_1 \bs_0 \bs_1 \bs_0\\
\bs_N \bs_{N-1} \bs_N \bs_{N-1} &= \bs_{N-1} \bs_N \bs_{N-1}\bs_N
\end{align}
The finite Weyl group $\cW_N^{\pp 0}$ of type $\cC_N$ is the subgroup of
$\cW_N$ generated by elements $s_1,\dots s_{N}$.
The affine Weyl group $\cW_N$ has a faithful action, depending on a parameter $q$, on
$\bC[z_1^\pm,\dots,z_N^\pm]$, the space of 
Laurent Polynomials in $N$ variables
\begin{align}
&\bs_i f(z_1,\dots,z_i,z_{i+1},\dots,z_N)= f(z_1,\dots,z_{i+1},z_{i},\dots,z_N)\\
&\bs_0 f(z_1,\dots,z_i,z_{i+1},\dots,z_N)= f(qz_1^{-1},\dots,z_i,z_{i+1},\dots,z_N)\\ 
&\bs_N f(z_1,\dots,z_i,z_{i+1},\dots,z_N)=
f(z_1,\dots,z_i,z_{i+1},\dots,z_N^{-1})
\end{align}
It has also actions on $\bZ^N$, parametrized by the level
$\ell\in \bZ$, 
\begin{align}
&\bs_i \{\alpha_1,\dots,\alpha_i,\alpha_{i+1},\dots,\alpha_N\}=
  \{\alpha_1,\dots,\alpha_{i+1},\alpha_{i},\dots,\alpha_N\}\\ 
&\bs_0 \{\alpha_1,\dots,\alpha_i,\alpha_{i+1},\dots,\alpha_N\}=
          \{\ell-\alpha_1,\dots,\alpha_i,\alpha_{i+1},\dots,\alpha_N\}\\  
&\bs_N \{\alpha_1,\dots,\alpha_i,\alpha_{i+1},\dots,\alpha_N\}=
\{\alpha_1,\dots,\alpha_i,\alpha_{i+1},\dots,-\alpha_N\}.
\end{align}
Notice that the action of $\cW_N^{\pp 0}$ on
$\bC[z_1^\pm,\dots,z_N^\pm]$ can be read off from its action on $\bZ^N$
by looking at monomials, namely for $w\in \cW_N^{\pp 0}$
$$
w \bz^\alpha = \bz^{w \alpha}.
$$
We define two partial orders on $\bZ^N$. The first is the usual
\emph{dominance order}: $\alpha \leq \beta$  if for  $1\leq j\leq N$
we have $\sum_{i=1}^j (\alpha_i-\beta_i)\leq 0$. 
The second order, $\alpha \preceq \beta$, is defined as follows.  
Call $\alpha^+$ the unique element in $\cW_N^{\pp 0}\alpha$ such that
$\alpha^+$ is a partition i.e. 
such that $\alpha_i\geq \alpha_{i+1}\geq 0$.
Then we say that  $\alpha \preceq \beta$ if $\alpha^+ < \beta^+$ or
in case $\alpha^+ = \beta^+$ if $\alpha \leq \beta$. 

\vskip .3cm
The Affine Hecke algebra $\cH_N$ of type $\tilde \cC_N$ is a deformation of the
group algebra of $\cW_N$, which depends on three parameters $t_0,t_N,t$ and 
is generated by elements $T_0,T_1,\dots,T_N$ subject to the
commutations relations
\begin{align}\label{hecke-def}
T_i-T_i^{-1}&= t_i^{\fr12}-t_i^{-\fr12}\\
T_iT_j&=T_jT_i ~~\textrm{if}~~|i-j|>1\\ 
T_i T_{i+1} T_i &= T_{i+1} T_i T_{i+1}~~\textrm{if}~~i\neq 0,N-1\\
T_0 T_1 T_0 T_1 &= T_1 T_0 T_1 T_0\\
T_N T_{N-1} T_N T_{N-1} &= T_{N-1} T_N T_{N-1}T_N
\end{align}
with $t_1=t_2=\dots=t_{N-1}=t$.

The finite Hecke algebra $\cH_N^{\pp 0}$ of type $\cC_N$ is a sub-algebra of $\cH_N$
generated by elements $T_1,\dots,T_{N-1},T_N$. 
It is well known \cite{humphreys1992reflection} that a basis of
$\cH_N$ is parametrized by elements of 
$\cW_N$ 
$$
T_w = T_{i_1}T_{i_2}\dots T_{i_\ell}
$$
where $w= s_{i_1}s_{i_2}\dots s_{i_\ell}$ is a reduced expression of
$w\in \cW_N$.

An important commutative sub-algebra $\cY_N$ is generated by Lusztig elements $Y_1^{\pm
1}, \dots, Y_N^{\pm 1} $ \cite{lusztig1989affine}
\be\label{lusztig}
Y_i = (T_i\dots T_{N-1})(T_N \dots T_0)(T_{1}^{-1}\dots T_{i-1}^{-1}).
\ee

In \cite{noumi1995macdonald}, Noumi introduced a representation of $\cH_N$ depending on
$6$ parameters $a,b,c,d, t,q$, acting on $\bC[z_1^{\pm
    1},\dots,z_N^{\pm 1}] $, the space of Laurent polynomials in $N$
variables $\bz=\{z_1,\dots,z_N\}$\footnote{Here we shall just consider
$a,b,c,d,t,q \in \bC$.} 
\begin{align}
\widehat T_i&=t^{\fr12} -(t^{\fr12}z_i-t^{-\fr12}z_{i+1})~ \partial_{i} \\
\widehat T_0 &=
t_0^{\fr12}-t_0^{-\fr12} \frac{(z_1-a)(z_1-b)}{z_1}~\partial_0\\
\widehat T_N &= t_N^{\fr12}-t_N^{-\fr12} \frac{(c z_N-1)(d z_N
  -1)}{z_N}\partial_N,
\end{align}
where $t_0=-q^{-1}ab$ $t_N=-cd$ and the finite difference operator
$\partial_{i}$, $\partial_0$ and $\partial_N$  are defined by
$$
\partial_{i} = \frac{1-{\bf s}_{i}}{z_i-z_{i+1}},~~~\partial_0 = 
\frac{1-{\bf s}_0}{z_1-qz_1^{-1}},~~~\partial_N = 
\frac{1-{\bf s}_N}{z_N-z_N^{-1}}. 
$$

\subsection{Analysis of eqs.(\ref{exch-open1}-\ref{exch-open3})}

Let us write the wave function $\Psi_{N,m}(\bz)$ in the basis
$\bw=w_1\dots w_N$, $w_i=\circ,\ast,\bullet$
$$
\Psi_{N,m}(\bz)=\sum_{\bw \in \cQ(N,m)} \psi_\bw(\bz)~ \bw,
$$
where the sum runs on $\cQ(N,m)$, the set of 2ASEP configurations on a
strip of
length $N$ in the presence of $m$ second class particles.
We can identify a word in $\circ,\ast,\bullet$ with an element of
$\bZ^N$ through the rules $\circ \equiv -1, \ast \equiv 0, \bullet
\equiv +1$, hence the set $\cQ(N,m)$ 
inherits an action of $\cW_N^{\pp 0}$ and of $\cW_N$ at level $0$ and
the partial orders $\leq, \preceq$.

Once written in components, the exchange-reflection 
eqs.(\ref{exch-open1}-\ref{exch-open3}) are expressed in a nice
compact form using 
the generators of the affine Hecke algebra $\cH_N$ in the Noumi
representation at $q=1$ 
\begin{align}\label{exch-comp1}
&\psi_w = t_i^{-\fr12}\widehat T_i\pp\psi_w &\textrm{if}&& &w=s_iw \\\label{exch-comp2}
&\psi_{s_i w} = t_i^{\fr12}\widehat T_i\pp\psi_w &\textrm{if}&& & \left\{
\begin{array}{ccc}
i>0 & \& & w < s_iw\\
i=0 & \& & w > s_0w  
\end{array}\right.
\end{align}
Recall that for $1\leq i \leq N-1$, $t_i=t$, while $t_0=-ab$ and $t_N=-cd$.

In words, eq.(\ref{exch-comp1}) tells us that if $w_i=w_{i+1}$ then
$\psi_w$ is symmetric under exchange $z_i\leftrightarrow z_{i+1}$,
while if $w_0=\ast$ or $w_N=\ast$, then $\psi_w$ is respectively
symmetric under inversions $z_1\rightarrow z_1^{-1}$ or
$z_N\rightarrow z_N^{-1}$.
Eq.(\ref{exch-comp2}) tells us how to exchange two different
neighboring particles 
\begin{equation*}
\begin{array}{ccc}
\cdots  \bullet \circ \cdots &  & \cdots \circ \bullet\cdots\\[-1pt]
\cdots  \bullet \ast \cdots & \xrightarrow{t^{\fr12}\widehat T_i} &\cdots \ast \bullet \cdots\\[7pt]
\cdots  \ast \circ \cdots & &\cdots \circ \ast  \cdots
\end{array}
\end{equation*}
and how to inject/remove particles from the boundaries.
\begin{align*}
\circ \cdots \xrightarrow{t_0^{\fr12}\widehat T_0} \bullet \cdots &&
\cdots \bullet  \xrightarrow{t_N^{\fr12}\widehat T_N} \cdots\circ
\end{align*}

\subsection{Reference component}

Let us consider a system of length $N=k+m$, in the sector with $m$
second class particles, and analyze the component associated to the state 
$$
\circ^k \ast^m
=\{\underbrace{\circ \dots
    \circ}_{k}\underbrace{\ast \dots \ast}_{m}\}
$$ 
It follows from the eqs.(\ref{exch-comp1},\ref{exch-comp2}), that
$\psi_{\circ^k \ast^m}(\bz)$ is preserved by the Lusztig operators $Y_i$.
Therefore (see Appendix \ref{algebraic}) we can identify 
$\psi_{\circ^k \ast^m}(\bz)$ with the non-symmetric Koornwinder
\be
\psi_{\circ^k \ast^m}(\bz) = E_{\mu(k,m)}(\bz)
\ee
associated to the string
$$
\mu(k,m)= \{\underbrace{-1,\dots,-1}_k,\underbrace{0,\dots,0}_m\}.
$$
This fixes completely the normalization of $\Psi_{N,m}(\bz)$.

The non-symmetric Koornwinder polynomials depend on the further
parameter $q$ which appears in the general Noumi representation. Hence
a priory we should specify that $E_{\mu(k,m)}(\bz)$ is taken for
$q=1$, but actually we will see in a moment that $E_{\mu(k,m)}(\bz)$
doesn't depend on $q$. 

Moreover, while the monomials expansion of
any non-symmetric Koornwinder polynomial reads as in
eq.(\ref{triangular}), the expansion of $E_{\mu(k,m)}(\bz)$ is much
simpler, it contains only negative powers of $z_i$ for $1\leq i \leq
k$. 

Either from the exchange eq.(\ref{exch-comp1}) applied to
$\psi_{\circ^k \ast^m}(\bz)$ or from 
the general theory of non-symmetric Koornwinder
polynomials \cite{sahi1999nonsymmetric}
we know that $E_{\mu(k,m)}(\bz)$ is separately symmetric in the first
$k$ and last $m$ spectral parameters, so that we can express it as a
sum of elementary symmetric polynomials in $z$ variables
\be\label{koorn-exp1}
E_{\mu(k,m)}(\bz) = \sum_{i=0}^k h_i^{(k,m)}
~ e_{k-i}(z_1^{-1}, \dots, z_k^{-1})
\ee
with $h_0^{(k,m)}=1$. Using the symmetry properties of
$E_{\mu(k,m)}(\bz)$ we can rewrite its defining equation (\ref{def-nonsymm-mcd})
as
\be\label{eig-eq-Ekm}
\widehat T_{k}\widehat T_{k+1}\dots \widehat T_{N-1}\widehat T_N
\widehat T_{N-1} \dots \widehat T_1\widehat T_0E_{\mu(k,m)}(\bz) = E_{\mu(k,m)}(\bz)
\ee
Pictorially this equation has a clear meaning: the action of $\widehat T_0$
injects a particle on the left, then $\widehat T_{N-1}\dots \widehat T_1$ moves this
particle to the far right, $\widehat T_N$ removes the particle from the right
and finally $\widehat T_{k}\widehat T_{k+1}\dots \widehat T_{N-1}$ brings back the empty
state in position $k$.

Equation (\ref{eig-eq-Ekm}) allows to determine $h_j^{(k,m)}$ in a
recursive way. It turns out that they do not depend on $k$ and have a
rather simple dependence on $m$. This is the content of the following
\begin{proposition}\label{ref-comp-prop1}
The non symmetric Koornwinder polynomial $E_{\mu(k,m)}(\bz)$ does not
depend on $q$ and reads
\be\label{int-form-non-symm}
E_{\mu(k,m)}(\bz) = \oint_0 \frac{dw}{2\pi i w}
H(w;a,b,t^mc, t^m d) \prod_{i=1}^k (w^{-1}+z_i^{-1})
\ee
where the integration is around zero and $H(w;a,b,c, d)$ is a formal
power series in $w$
\be
H(w;a,b, c, d) =\sum_{j=0}^\infty h_n(a,b,c,d)w^n
\ee
with
\begin{multline}\label{sol-h}
h_n(a,b,c,d)= \\\frac{1}{(abcd;t)_n}\sum_{\substack{0\leq
    i,j,k\leq n\\i+j+k+h\leq n}}{n \brack i,j,k}_t
t^{\binom{i}{2}+\binom{j}{2}}a^{i}b^{j}(-c)^{n-k}(-d)^{i+j+k}
\end{multline}
where we have used the usual q-multinomial with base $t$ 
$$
{n \brack i,j,k}_t=
\frac{(t;t)_n}{(t;t)_i(t;t)_j(t;t)_k(t;t)_{n-i-j-k}},~~~(a;t)_n=\prod_{i=0}^{n-1}(1-at^i). 
$$
\end{proposition}
\begin{proof}
Let us rewrite eq.(\ref{koorn-exp1}) as a contour
integral around the origin
\be\label{contour1}
E_{\mu(k,m)}(\bz) = \oint_0 \frac{dw}{2\pi i w}H_{k,m}(w)
\prod_{i=1}^k (w^{-1}+z_i^{-1})
\ee
with $H_{k,m}(w)=\sum_{j=0}^kh_j^{(k,m)}w^{j}$.
Using the following relations 
\begin{align*}
\widehat T_0 (w^{-1}+z_1^{-1}) &=
q^{-1}t_0^{-\fr12}\left((w+z_1)-\frac{(w+a)(w+b)}{w}\right) \\
\widehat T_j(w+z_j)(w^{-1}+z_{j+1}^{-1})&=t^{-\fr12}(w^{-1}+z_{j}^{-1})(w+z_{j+1})\\
\widehat T_j(w^{-1}+z_{j+1}^{-1})&=t^{-\fr12}(tw^{-1}+z_{j}^{-1})\\
\widehat T_j(w+z_{j})&=  t^{-\fr12}(tw+z_{j+1})\\
\widehat T_N (w+z_N) &=
t_N^{-\fr12}\left((w^{-1}+z_N^{-1})-\frac{(1+wc)(1+wd)}{w}\right) 
\end{align*}
in eq.(\ref{eig-eq-Ekm}) we get
$$
t_N^{\fr12}t^m\oint_0 \frac{dw}{2\pi
  i}H_{k,m}(w)\frac{(a+w)(b+w)}{w}\prod_{j=1}^{k-1}(tw^{-1}+z_j^{-1})= 
$$
$$
t_N^{-\fr12}\oint_0 \frac{dw}{2\pi
  i}H_{k,m}(w)\frac{(1+t^mcw)(1+t^md w)}{w}\prod_{j=1}^{k-1}(w^{-1}+z_j^{-1}) 
$$
The fact that this equation must be true for any choice of
$z_1,\dots,z_{k-1}$  
means that we can substitute $\prod_{j}^{k-1}(w^{-1}+z_j^{-1})$ with any polynomial $p(w^{-1})$
of degree $k-1$, getting
\be\label{generat-eq-h}
\begin{split}
\oint_0 \frac{dw}{2\pi
  i}H_{k,m}(w)\left(p(tw^{-1})\frac{(a+w)(b+w)}{w}\right) =\\
\oint_0 \frac{dw}{2\pi
  i}H_{k,m}(w)\left(p(w^{-1})
\frac{(w+t^{-m}c^{-1})(w+t^{-m}d^{-1})}{w} \right) 
\end{split}
\ee
From the previous equation we can already make a couple of
conclusions.\\ 
1) Since in eq.(\ref{generat-eq-h}) the parameter $q$ (which was
present through $t_0=-q^{-1}ab$) has disappeared
we have that $E_{\mu(k,m)}(\bz) $ does not depend on $q$.\\
2) The dependence on $k$ is only in the degree of the arbitrary
polynomial $p(w)$, which means that the coefficients  $h_j^{(k,m)}$
actually do not depend on $k$, therefore we can suppress the
label $k$ and we think of $H_{k,m}(w)=H_{m}(w) $ as a
formal infinite series. \\
3) The parameter $m$ appears just by multiplying
both $c$ and $d$ by $t^m$ and  we can write the solution
of eq.(\ref{generat-eq-h}) as
$$
H_{m}(w;a,b,c,d)=H(w;a,b,t^mc, t^m d)
$$
and 
$$
h_j^{(m)}=h_j(a,b,t^m c,t^m d) 
$$
It remains only to use eq.(\ref{generat-eq-h}) to determine the
coefficients $h_j(a,b,c,d)$. 
By setting $p(w)=w^{n-2}$ into eq.(\ref{generat-eq-h}) we get
following three-terms recursion relation 
\begin{multline}\label{rec-h_i}
(t^{n-1}ab-c^{-1}d^{-1})h_n+\\ 
(t^{n-1}(a+b)-(c^{-1}+d^{-1}))h_{n-1}
+ (t^{n-1}-1)h_{n-2}=0.
\end{multline}
with initial conditions $h_{0}=1$. 
It is then not difficult to check that the expression in
eq.(\ref{sol-h}) satisfies this recursion relation.
\end{proof}
An immediate consequence of Proposition \ref{ref-comp-prop1} is that 
\be\label{recursion-psi-ref}
\psi_{\circ^k \ast^m}(\bz;a,b,c,d)=\psi_{\circ^k\ast^{m-1}}(\bz_{\widehat N};a,b ,tc
,td).
\ee
where $\bz_{\widehat j}$ means that the variable $z_j$ is absent. 

The recursion relation for the coefficients $h_i$ eq.(\ref{rec-h_i})
has some invariance properties. It is
invariant under $a\leftrightarrow c^{-1},b\leftrightarrow
d^{-1},t\leftrightarrow t^{-1}$, from which (writing explicitly  the
dependence of $h_n$ on $t$ as $h_n(a,b,c,d|t)$) we  get the duality relation
\be\label{inv-h_i}
h_n(a,b,c,d|t)=h_n(c^{-1},d^{-1},a^{-1},b^{-1}|t^{-1}).
\ee
It is also invariant under 
\begin{align*}
a\rightarrow~\lambda^{-1}a,&&
b\rightarrow~\lambda^{-1}b,&&  c\rightarrow~\lambda~c,&& d\rightarrow~\lambda~d,&&
h_n\rightarrow~\lambda^n~h_n,
\end{align*}
which means that
$h_n(a,b,c,d)$ is homogeneous in
$a^{-1},b^{-1},c,d$ of degree $n$ and we have
\be\label{H-hom}
H(\lambda
^{-1}w;\lambda^{-1} a,\lambda^{-1} b,\lambda c, \lambda d) = H(w;a,b,c, d).
\ee
Using repeatedly the relation
$t^{-\fr12}\widehat T_j^{-1}(w^{-1}+z_{j}^{-1})=(t^{-1}w^{-1}+z_{j+1}^{-1})$ 
we obtain easily the components of all the configurations without
first class particles. For $J=\{j_1,\dots,j_k\}$, $1\leq j_1< j_2<\cdots j_k\leq N$ call
$ \circ(J)_N $ the configuration with empty space in positions $J$ and
second class particles in all the other sites. Then we have
$$
\psi_{\circ(J)_N}(\bz)=\oint_0 \frac{dw}{2\pi i w}
H(w;a,b,t^mc, t^m d) \prod_{i=1}^k (t^{i-j_i}w^{-1}+z_{j_i}^{-1}). 
$$
For configurations of the kind $\ast^m \circ^k $, combining the
previous equation with eq.(\ref{H-hom}) we
  obtain the analog of eq.(\ref{recursion-psi-ref})  
\be\label{recursion-psi-ref2}
\psi_{ \ast^m\circ^k}(\bz;a,b,c,d)=\psi_{\ast^{m-1}\circ^k}(\bz_{\hat 1};t^{-1}a,t^{-1}b ,c
,d).
\ee
Analogous integral formulas can be written for the components of the
configurations without empty sites.
For later reference
we write explicitly $h_1$
\be\label{h1}
h_1(a,b,c,d)= \frac{a+b-c^{-1}-d^{-1}}{c^{-1}d^{-1}-ab}.
\ee

Another interesting remark is that from a close inspection it turns
out that the three-terms recursion
(\ref{rec-h_i}) is nothing else than the 
recursion relation for Al-Salam Chihara polynomials in disguise (see
eq.(\ref{al-salam-chi-rec})) therefore we 
have also
\be
h_n(a,b,c,d) =
\frac{(-1)^nc^{\frac{n}{2}}d^{\frac{n}{2}}}{(abcd;t)_n}Q_n\left(\frac{1}{2}
\left(\sqrt{\frac{c}{d}}+\sqrt{\frac{d}{c}}\right)  
;c^\fr12 d^\fr12 a, c^\fr12 d^\fr12b|t\right).
\ee

\subsection{Other components and normalization}

Once the component $\psi_{\circ^k \ast^m}$ is known, all the others are
obtained from it by action of the generators $\widehat T_i$ using the exchange
equations (\ref{exch-comp2}). Well, actually it is not necessary to
use the full algebra $\cH_N$, $\cH_N^{\pp 0}$ being enough. 
Moreover the state $\psi_{\circ^k \ast^m}$ is preserved by the generators
$\widehat T_i$ for $1\leq i\neq N-m~\leq~N$.
In order to see how to remove this redundancy and to provide a compact
formula for any component $\psi_w(\bz)$ we need to concentrate for a
moment on the Weyl group $\cW_N^{\pp 0}$ and its action on $\cQ(N,m)$.  
The state $w(N-m,m)$ is preserved by $I:=\cS_{N-m}\times \cW^{\pp 0}_m$,
which is a parabolic subgroup of $\cW_N^{\pp 0}$ i.e. a subgroup generated
by a subset of the generators of $\cW_N^{\pp 0}$. In the present case $I$ is
generated by ${\bf s}_i$ for $1\leq i \neq N-m \leq N$.
The states in $\cQ(N,m)$ are in bijection with elements in $\cW_N^{\pp 0}/I$.
On the other hand, from the basic properties of parabolic subgroups of a
Coxeter group (see \cite{humphreys1992reflection} Chapter 1) 
we know that the set 
$\cW^I$ defined 
$$
\cW^I = \{{\bf g}\in \cW^0_N| \ell({\bf gs})> \ell({\bf g}), \forall
   {\bf s}\in I\}.
$$
is in bijection with classes in the
quotient $\cW_N^{\pp 0}/I$, or equivalently, that any classes in
$[{\bf g}]\in\cW_N^{\pp 0}/I$ 
 contains a unique shortest representative ${\bf h}\in\cW^I$  and 
 therefore any ${\bf g}\in \cW_N^{\pp 0}$ can be 
written in a unique way as a product ${\bf g}={\bf hs}$ with ${\bf h}\in \cW^I$  and
${\bf s}\in I$. Moreover, any reduced expression of ${\bf g}$ is the
product of a reduced expression of ${\bf h}$ times a reduced expression of
${\bf s}$, which in particular implies that at the level of the Hecke algebra, we have a 
decomposition
\be\label{decomp-hs}
T_{\bf g}= T_{\bf h} T_{\bf s}
\ee
with ${\bf h}\in W^I$ and ${\bf s}\in I$, and 
\be\label{decomp-chi-hs}
\chi(T_{\bf g})=\chi(T_{\bf h})\chi(T_{\bf s}).
\ee
The bijection between $W^I$ and $\cQ(N,m)$ is now obvious
$$
{\bf h} \mapsto {\bf h}\circ^k \ast^m.
$$
Coming back to the components we easily check that for ${\bf h}\in
W^I$ we have
\be
\psi_{{\bf h}\circ^k \ast^m}(\bz) = \chi(T_{\bf h})T_{{\bf h}}  \psi_{\circ^k \ast^m}(\bz).
\ee

It is now easy to prove the following theorem that identify the
normalization of $\Psi_{N,m}(\bz)$ (sometimes called also the partition function) with a
Macdonald-Koornwinder polynomial \cite{koornwinder1992askey}
\begin{theorem}\label{theo-mac-koorn}
Let 
\be
Z_{N,m}(\bz):=\sum_{w\in \cQ(N,m)}\psi_w(\bz)
\ee
Then we have
\begin{equation}\label{int-form-MD-Koornw}
Z_{N,m}(\bz)= P_{\eta(N,m)}(\bz).
\end{equation}
where $P_{\eta(N,m)}(\bz)$ is the Macdonald-Koornwinder polynomial \cite{koornwinder1992askey}
associated to the partition
$$
\eta(N,m)=\{\underbrace{1,\dots,1}_{N-m},\underbrace{0,\dots,0}_{m} \} .
$$
\end{theorem}
\begin{proof}
From \cite{sahi1999nonsymmetric} (see Theorem \ref{theo-sahi} in
Appendix \ref{algebraic})  we know that
$$
P_{\eta(N,m)}(\bz) \propto \sum_{{\bf g}\in \cW_N^{\pp 0}} \chi(T_{\bf g})T_{\bf g} \psi_{\circ^k \ast^m}(\bz).
$$
Now rewrite the sum over ${\bf g}\in \cW_N^{\pp 0}$ as a double sum over 
${\bf h}\in \cW^I$ and  over ${\bf s}\in I$, using eqs.(\ref{decomp-hs},\ref{decomp-chi-hs})
$$
P_{\eta(N,m)}(\bz) \propto \sum_{{\bf h}\in \cW^I} \sum_{{\bf s}\in
  I}\chi(T_{\bf h})\chi(T_{\bf s})T_{\bf h}T_{\bf s}
\psi_{\circ^k \ast^m}(\bz). 
$$
For ${\bf s}\in I$ one has that $T_{\bf
  s}\psi_{\circ^k \ast^m}(\bz)\propto \psi_{\circ^k \ast^m}(\bz)$ and hence
$$
\sum_{{\bf s}\in
  I}\chi(T_{\bf s})T_{\bf s}\psi_{\circ^k \ast^m}(\bz)\propto \psi_{\circ^k \ast^m}(\bz).
$$
Therefore we obtain
$$
P_{\eta(N,m)}(\bz) \propto \sum_{{\bf h}\in \cW^I}\chi(T_{\bf
  h})T_{\bf h}\psi_{\circ^k \ast^m}(\bz) = \sum_{w\in \cQ(N,m)} \psi_w(\bz).
$$
In order to fix the proportionality constant we recall that
$P_{\eta(N,m)}(\bz)$ is normalized in such a way that the monomial
$\bz^{\mu(N-m,m)}$  has coefficient $1$. On the other side, in the sum
defining $Z_{N,m}(\bz)$, such monomial  comes
only from the component $\psi_{\circ^k \ast^m}(\bz)$ and has coefficient
$1$. Therefore we conclude that the proportionality factor is $1$.
\end{proof}
Actually we can do even better and consider the generating function
for the number of particles in the system. 
Calling $\bullet(w)$ the number of particles present in the state
$w$, what we would like to consider is the weighted partition function
\be
Z_N(\xi;\bz):= \sum_{w\in \cQ(N,m)} \xi^{\bullet(w)}\psi_w(\bz).
\ee
In order to state our result it is convenient to switch to a notation
in which the dependence on the parameters $a,b,c,d$ is explicit
writing  $Z_{N,m}(\bz;a,b,c,d)$ for the normalization. Then we have
\begin{theorem}\label{gener-funct}
The generating function for the number of particles in the system is
given by
\be
Z_{N,m}(\xi^2;\bz;a,b,c,d) = \xi^{N-m}Z_{N,m}(\xi\bz; a_\xi, b_\xi,c_\xi,d_\xi).
\ee
with $a_\xi=\xi~a, b_\xi= \xi~b, c_\xi=\xi^{-1}c, d_\xi=\xi^{-1}d$.
\end{theorem}
This theorem is just an immediate corollary of the following
\begin{proposition}
For any state $w\in \cQ(N,m)$ we have
\be\label{weigth-psi}
\xi^{2\bullet(w)}\psi_w(\bz;a,b,c,d) = \xi^{N-m} \psi_w(\xi\bz; a_\xi,
b_\xi,c_\xi,d_\xi). 
\ee
\end{proposition}
\begin{proof}
For $w=\circ^k \ast^m$ eq.(\ref{weigth-psi}) is equivalent to
eq.(\ref{H-hom}).
Then we claim that if
eq.(\ref{rec-h_i}) is true for some $w$ then it is true for all the
states with the same number of first particles. Indeed if
$\bullet(w)=\bullet(w')$, then $\psi_{w'}$ is obtained from $\psi_w$
through the action of the sole operators $T_i$ with $1\leq i \leq N-1$,
which are homogeneous in the variables $\bz$. 

The last step is an induction on the number of first class particles. Assume
eq.(\ref{rec-h_i}) is true for all $w$ such that $\bullet(w)\leq 
k$. Then consider a state $\bar w$ with $\bullet(\bar w)=n$, of the form
$\bar w=\tilde w \circ$. Writing $\psi_{\tilde w\circ }(\bz)$ as
$$
\psi_{\tilde w\circ}(\bz)=z_N^{-1}\phi^{(-1)}(\bz_{\widehat N})+\phi^{(0)}(\bz_{\widehat N}),
$$ 
the relation (\ref{rec-h_i}) translates into
\be
\begin{split}\label{homogeneity}
\xi^k\phi^{(0)}(\bz_{\widehat N};a,b,c,d) &= \xi^{N-m} \phi^{(0)}(\xi \bz_{\widehat N};a_\xi,
b_\xi,c_\xi,d_\xi)\\
\xi^k\phi^{(1)}(\bz_{\widehat N};a,b,c,d) &= \xi^{N-m-1}
\phi^{(1)}(\xi\bz_{\widehat N}, a_\xi, 
b_\xi,c_\xi,d_\xi).
\end{split}
\ee
Acting  with $t_N^{-\fr12}T_N^{-1}$ on the component $\psi_{\tilde
  w\circ}(\bz)$ we get 
$$
\psi_{\tilde w\bullet}(\bz) =
(z_N-c^{-1}-d^{-1})\phi^{(-1)}(\hat\bz)-c^{-1}d^{-1}\phi^{(0)}(\hat\bz).  
$$
Using this expression and eqs.(\ref{homogeneity}) we easily verify
that $\psi_{\tilde w\bullet}(\bz)$ satisfies eq.(\ref{rec-h_i}).
\end{proof}

A more explicit formula for the weighted partition function
$Z_{N,m}({\xi^2;\bf z})$, which is well suited for asymptotic
analysis, is obtained in terms of an integral
representation of the Macdonald-Koornwinder polynomials \cite{mimachi2001duality}
(see Appendix
\ref{symm-macd}). Assuming $t<1$ we can write
\begin{multline}
\label{integral-Z}
Z_{N,m}(\xi^2;{\bf z}) =r_m^{-1}(a_\xi,b_\xi,c_\xi,d_\xi|t) \times \\ 
\xi^{N-m}\oint_\cC
\frac{dx}{4\pi i x} \Pi(\bz,x) w(x;a_\xi,b_\xi,c_\xi,d_\xi|t) p_{m}(x;a_\xi,b_\xi,c_\xi,d_\xi|t) 
\end{multline}
where $p_{m}(x;a,b,c,d|t)$  is the $m$-th Askey-Wilson polynomial of
base $t$ in the variable $\frac{x+x^{-1}}{2}$, 
 $w(x;a,b,c,d|t)$   is the Askey-Wilson kernel 
\be
w(x;a,b,c,d|t)=\frac{(x^2,x^{-2};t)_\infty}{
  (ax,ax^{-1},bx,bx^{-1},cx,cx^{-1},dx,dx^{-1};t)_\infty},
\ee
while $\Pi(\bz,\bx) = \prod_{1\leq i \leq N}
  (z_i+z_i^{-1}-x-x^{-1})$ and the normalization constant $r_m$ is given by
\be
r_m(a,b,c,d|t)=
\frac{(abcdt^{2m};t)_\infty}{(t^{m+1},abt^m,act^m,adt^m,bct^m,bdt^m,cdt^m;t)_\infty}.  
\ee 
The contour of integration $\cC$ encircles the poles in
$a_\xi t^k,b_\xi t^k,c_\xi t^k,d_\xi t^k$ ($k\in \bZ_+$) and excludes all the others. 

In the homogeneous limit, our result for $Z_{N,m}(\xi;{\bf z})$ looks
quite different and simpler than the analogous result in
\cite{uchiyama2008two}, where Uchiyama computes the grand partition
function. 
A more direct comparison can be done for $m=0$ with the expression for
$Z_{N,m}(\xi;{\bf 1})$ presented in
\cite[Eq.(6.2)]{uchiyama2004asymmetric} (their parameters $a,b,c,d,q$
are related to ours by $q=t^{-1}, a\leftrightarrow -a^{-1}, d\leftrightarrow
-d^{-1},b\leftrightarrow-c^{-1}$). Since $p_0(x|t)=1$,  we see that the integral
in eq.(\ref{integral-Z}) coincides with the one in
\cite[Eq.(6.2)]{uchiyama2004asymmetric}, nonetheless the two formulas
differ by a prefactor which depends on $\xi$ (and therefore cannot be
accounted by a different normalization of the wave function), but not on
the size $N$ of the system. This gives
rise at finite size to different predictions for the average
occupation number of first class particles.

\section{Recursion relations}\label{sect-rec-rel}

In this Section we will discuss some relations between solutions of the
exchange-reflection equations for systems of different size and different values
of the parameters $a,b,c,d$. This will allow us to derive some
contiguous relations for Askey-Wilson polynomials. 
Whenever we shall compare solutions corresponding to different values of
the parameter $a,b,c,d$ we shall keep
explicit the dependence on these parameters.
   
\begin{theorem}\label{recursion-theo}
The following three kinds of recursive relations holds
\be\label{recurs1}
\begin{split}
\psi_{w\ast}(\bz;a,b,c,d) &= \psi_{w}(\bz_{\widehat N};a,b,tc,td)\\
\psi_{\ast w}(\bz;a,b,c,d) &= \psi_{w}(\bz_{\hat 1};ta,tb,c,d).
\end{split}
\ee
\be\label{recurs2}
\begin{split}
\psi_{w\circ}(\bz) + c
d\psi_{w\bullet}(\bz)&= \frac{(1-cz_N)(1-dz_N)}{z_N}
\psi_{w}(\bz_{\widehat N})\\
ab\psi_{\circ w}(\bz) + 
\psi_{\bullet w}(\bz)&= \frac{(a-z_1)(b-z_1)}{z_1}
\psi_{w}(\bz_{\hat 1})
\end{split}
\ee
\be\label{recurs3}
\begin{split}
\psi_{w\circ}(\bz;a,b,c,d)|_{z_N=c^{-1}}&=K_R(c)
\psi_{w}(\bz_{\widehat N};a,b,tc,d)\\
\psi_{w\circ}(\bz;a,b,c,d)|_{z_N=d^{-1}}&=K_R(d)
\psi_{w}(\bz_{\widehat  N};a,b,c,td)\\
\psi_{w\bullet}(\bz;a,b,c,d)|_{z_N=c^{-1}}&=-c^{-1}d^{-1}K_R(c)
\psi_{w}(\bz_{\widehat N};a,b,tc,d)\\
\psi_{w\bullet}(\bz;a,b,c,d)|_{z_N=d^{-1}}&=-c^{-1}d^{-1}K_R(d)
\psi_{w}(\bz_{\widehat  N};a,b,c,td)
\end{split}
\ee
\be\label{recurs4}
\begin{split}
\psi_{\circ w}(a,b,c,d;\bz)|_{z_1=a}&=K_L(a)
\psi_{w}(ta,b,c,d;\bz_{\hat 1})\\
\psi_{\circ w}(a,b,c,d;\bz)|_{z_1=b}&=K_L(b)
\psi_{w}(a,tb,c,d;\bz_{\hat 1})\\
\psi_{\bullet w}(a,b,c,d;\bz)|_{z_1=a}&=-abK_L(a)
\psi_{w}(ta,b,c,d;\bz_{\hat 1})\\
\psi_{\bullet w}(a,b,c,d;\bz)|_{z_1=b}&=-abK_L(b)
\psi_{w}(a,tb,c,d;\bz_{\hat 1})
\end{split}
\ee
where $\bz_{\hat \ell}$ means that the variable $z_{\ell}$ is absent
and
\be\label{prop-const}
\begin{split}
K_R(x)&=-\frac{(1-axt^m)(1-bxt^m)cdx^{-1}}{1-abcdt^{2m}}\\
K_L(x)&=\frac{(1-cxt^m)(1-dxt^m)x^{-1}}{1-abcdt^{2m}}
\end{split}
\ee
\end{theorem}
\begin{proof}
First we claim that 
if eqs.(\ref{recurs1}-\ref{recurs4}) hold for some $w\in \cQ(N,m)$ then
they hold for any configurations in $\cQ(N,m)$. 
Indeed, for the equations in which the particle configuration at site
$N$ is fixed, 
one can modify $w$ by acting on both sides of the
equalities either with operators $T_i$ with $0\leq i<N-1$, while for the
equations in which it is the first site configuration to be fixed, one
can modify $w$ by acting with operators $T_i$ with $1< i\leq N$.

We shall use the previous remark to prove
eqs.(\ref{recurs1},\ref{recurs2}). We could use it also to prove
eqs.(\ref{recurs3},\ref{recurs4}) but we prefer to adopt a different strategy. 

The first equation of (\ref{recurs1}), in the case $w=\circ^k \ast^m$
coincides with eq.(\ref{recursion-psi-ref}), while the second equation
of  (\ref{recurs1}) for $w=\ast^m\circ^k $ coincides with
eq.(\ref{recursion-psi-ref2}). 

In order to prove the second of eqs.(\ref{recurs2}) for $w=\circ^{k-1} \ast^{m}$,
we use the integral expression of 
$\psi_{\circ^k \ast^{m}}(a,b,c,d;\bz)$  given by
eq.(\ref{int-form-non-symm}) 
\begin{gather*}
ab\psi_{\circ \circ^{k-1} \ast^{m}}(\bz)+\psi_{\bullet \circ^{k-1} \ast^{m}}(\bz) =\\
(ab+t_0^\fr12 T_0)\psi_{\circ^{k} \ast^{m}}(\bz)=\\
(ab+t_0^\fr12 T_0)\oint_0\frac{dw}{2\pi
  iw}F_m(w)\prod_{j=1}^{k}\left(w^{-1}+z_j^{-1}\right)=\\
 \frac{(a-z_1)(b-z_1)}{z_1}\oint_0\frac{dw}{2\pi
  iw}F_m(w)\prod_{j=2}^{k}\left(w^{-1}+z_j^{-1}\right)=\\
\frac{(a-z_1)(b-z_1)}{z_1}\psi_{\circ^{k-1} \ast^{m}}(\bz_{\hat 1})
\end{gather*}
where in the third equality we have used the identity
$$
(ab+ t_0^\fr12 T_0)(w^{-1}+z_1^{-1})=  \frac{(a-z_1)(b-z_1)}{z_1}.
$$
In a similar way one proves the first of eqs.(\ref{recurs2}).

For the proof of eqs.(\ref{recurs3}) we adopt a different strategy: we
show that the  l.h.s.  
satisfies the same exchange-reflection equations as the r.h.s. and
therefore by Theorem \ref{unicity} they must be proportional.

Lets look at the first of these equations, since we are comparing
systems of different length we must distinguish between
operators representing the Affine Hecke algebra of different rank and
different parameters therefore just for the present proof we shall use
an heavier notations. For 
the case of rank $N$ with parameter $c$ highlighted we shall write the
generators as $T_i^{(N)}(c)$. In terms of divided difference operators we have
of course $T_i^{(N)}(c) =  T_i^{(N-1)}(c')$ for $i<N-2$ and any $c,
c'$, therefore both
sides of the first of eqs.(\ref{recurs3}) satisfies
eqs.(\ref{exch-comp1},\ref{exch-comp2}) for 
$i<N-2$ and $T_i^{(N)}(tc)$. 
It remains to prove  that the r.h.s. satisfies
eqs.(\ref{exch-comp2}) for $i=N-1$. If $w= \tilde w \ast$ then
$\psi_{w\circ}$ does not depend on $z_{N-1}$ so in 
particular it is invariant under $z_{N-1}\leftrightarrow z_{N-1}^{-1}$.
If $w= \tilde w \circ$ then $\psi_{\tilde w \circ\circ}$ is symmetric in
$z_{N-1},z_N$ and we have   
\be\label{compos1}
\psi_{\tilde w \bullet\circ}(\bz) = \left((-tcd)^\fr12T_N^{(N)}(c)T_{N-1}^{(N)}(c)\right)^{-1}\psi_{\tilde w \circ\circ}(\bz)
\ee
It is straightforward to verify that for any symmetric function $H(x,y)$
one has the following identity
$$
 \left(T_N^{(N)}(c)T_{N-1}^{(N)}(c)\right)^{-1}H(z_{N-1},z_{N})|_{z_N=c^{-1}}
 =  \left( T_{N-1}^{(N-1)}(tc)\right)^{-1}H(z_{N-1},c^{-1}).
$$ 
This means that after specialization $z_N=c^{-1}$ in
eq.(\ref{compos1}) we get 
$$
\psi_{\tilde w \bullet\circ}(\bz)|_{z_N=c^{-1}}= \left(
T_{N-1}^{(N-1)}(tc)\right)^{-1} \psi_{\tilde w
  \circ\circ}(\bz)|_{z_N=c^{-1}}. 
$$
It remains to fix the proportionality constant $K_R(c)$. This 
is given by the coefficient of the monomial
$\prod_{j=1}^kz_j^{-1}$ in   
$\psi_{\ast^{m}\circ^{k}
  }(\bz;a,b,c,d)|_{z_N=c^{-1}}$, since  the same monomial
has coefficient $1$ in $\psi_{\ast^{m}\circ^{k-1}
  }(\bz_{\widehat N};a,b,tc,d)$. We get
$$
K_R(c) = c+ t^m h_1(a,b,t^m c,t^m d|t) =
\frac{(-1+act^m)(-1+bct^m)d}{-1+abcdt^{2m}}. 
$$ 
In a similar way one proves the identities for the specializations
$z_1=a,b$. 
\end{proof}
Since $\psi_{w\circ}(\bz;a,b,c,d)$  and $\psi_{w\bullet}(\bz;a,b,c,d)$
are polynomials of degree $1$ in $z_N$ and $z_N^{-1}$ respectively,
they are recursively determined by the two equations (\ref{recurs3}). 
At the level of partition function the recursions
(\ref{recurs1},\ref{recurs3}) imply
\begin{multline}\label{rec-part-func}
Z_{N,m}(\bz;a,b,c,d)|_{z_N=c^{-1}}= Z_{N-1,m-1}(\bz_{\widehat N};a,b,tc,td)+\\
\frac{(1-act^m)(1-bct^m)(1-dc)}{c(1-abcdt^{2m})}
Z_{N-1,m-1}(\bz_{\widehat N}a,b,tc,d).
\end{multline}
By using the integral formula (\ref{int-form-MD-Koornw}) in the
previous equation, we derive a contiguous relation for
Askey-Wilson polynomials which is not difficult to prove by direct
means\footnote{Several  contiguous/difference relations for
Askey-Wilson polynomials emerging from our discussion are presented in
Appendix \ref{AW-section}.}
\begin{multline}\label{cont12-mt}
(1-cdt^m)P_m(x;a,b,c,d|t)-(1-cd)P_m(x;a,b,tc,d|t)\\-cd(1-t^m)(1-abt^{m-1})(z+z^{-1}-d-d^{-1})
P_{m-1}(x;a,b,tc,td|t)=0. 
\end{multline}
On the other hand eq.(\ref{rec-part-func}) and its analog for
$z_N=d^{-1}$, determine $Z_{N,m}(\bz;a,b,c,d)$ uniquely. Therefore
assuming eq.(\ref{cont12-mt}), we could obtain a proof of the integral formula
(\ref{int-form-MD-Koornw}), without employing Mimachi's formula.

Another consequence of the recursions (\ref{recurs3},\ref{recurs4})
concerns the coefficients $h_n$ in the expansion of
$\psi_{\circ^k\ast^m}$. Indeed by applying
eqs.(\ref{recurs3},\ref{recurs4}) to the case $w=\circ^{n}$ we get 
$$
h_{n+1}(a,b,c,d)+c h_{n}(a,b,c,d)=-\frac{(1-ac)(1-bc)d}{(1-abcd)},
h_{n}(a,b,tc,d) 
$$
$$
h_{n+1}(a,b,c,d)+a^{-1} h_{n}(a,b,c,d)=\frac{(1-ac)(1-ad)a^{-1}}{1-abcd}h_{n}(ta,b,c,d).
$$
Actually these two equations are equivalent thanks to the duality
relation (\ref{inv-h_i}). The previous two equations can be written
as contiguous or difference equations for Al-Salam Chihara polynomials   
\begin{multline}
aQ_{n+1}(z;a,b|t)-(1-abt^n)Q_{n}(z;a,b|t)\\-a(z+z^{-1}-a-a^{-1})Q_{n}(z;ta,b|t)=0,
\end{multline}
\begin{multline}
Q_{n+1}(z;a,b|t)-z(1-abt^n)Q_{n}(z;a,b|t)\\-t^{n/2}\frac{(1-az)(1-bz)}{z}Q_{n}(t^\fr12
z;t^\fr12 a,t^\fr12 b|t)=0.
\end{multline}
Both equations can be lifted to contiguous or difference relations for Askey-Wilson
polynomials (see Appendix \ref{AW-section}).

\section{Density and current}\label{sect-phase}

In order to compute physical quantities, like average density of first
class particles or current, we need to specialize the spectral
parameters  $\bz= {\bf 1}$. 
While for the density this is straightforward 
\be\label{av-dens}
\langle \rho^\bullet_{N,m}\rangle = \left(\frac{1}{N}\frac{\partial}{\partial \xi} \log
Z_{N,m}(\xi;{\bf 1})\right){\Big |}_{\xi=0},
\ee
for the steady current $\langle J_{N,m}\rangle$ we need to make an intermediate step and define
an inhomogeneous version of that quantity
\be
J_{N,m}(\bz):= \frac{1}{Z_{N,m}(\bz)}\sum_{w\in \cQ(N,m)} \alpha
\psi_{\circ w}(\bz)-\gamma \psi_{\bullet w}(\bz),
\ee
which under specialization $\bz= {\bf 1}$ reduces to the homogeneous case,
$$
\langle J_{N,m}\rangle~=~J_{N,m}({\bf 1}).$$
Using the fact that $\alpha/\gamma=-ab$ and the first of
eqs.(\ref{recurs2}), we get an explicit formula for $J_{N,m}(\bz)$
\be
J_{N,m}(\bz) = -\gamma \frac{(a-z_1)(b-z_1)}{z_1}
\frac{Z_{N-1,m}(\bz\setminus z_1)}{Z_{N,m}(\bz)}. 
\ee
After specialization $\bz ={\bf 1}$, we obtain
\be\label{av-curr}
\langle J_{N,m} \rangle = (t^\fr12 -t^{-\fr12})
\frac{Z_{N-1,m}({\bf 1})}{Z_{N,m}({\bf 1})}.
\ee

\subsection{Thermodynamic limit}\label{asympt-1}

Since the stationary current
and the average first particle density have been expressed in terms of the weighted
partition function, in order to determine the behavior of
such quantities in the thermodynamic limit we need to work out the
asymptotic behavior of $Z_{N,m}(\xi^2;{\bf 1})$ when $N$ and $m$ tend
to infinity at fixed density of second class particles $\rho_\ast=m/N$.

If we assume $t<1$, we can use eq.(\ref{integral-Z}). In that equation
the prefactor $r_m^{-1}(a_\xi,b_\xi,c_\xi,d_\xi|t)$ for $m\rightarrow \infty$ goes to a
constant $\kappa$, therefore 
we just need to consider the
integral, which in the homogeneous specialization reads
\begin{multline}
\oint_\cC
\frac{dx}{4\pi i
  x}\frac{(x^2,x^{-2};t)_\infty (\xi+\xi^{-1}-x-x^{-1})^N p_{m}(x;a_\xi,b_\xi,c_\xi,d_\xi|t)}{(a_\xi
  x,a_\xi x^{-1},b_\xi x,b_\xi x^{-1},c_\xi x,c_\xi x^{-1},d_\xi
   x,d_\xi x^{-1};t)_\infty} 
\end{multline}

For the large $N,m$ we can use the asymptotic formula
eq.(\ref{asympt-AW}) for Askey-Wilson polynomials and we arrive at
\be
Z_{N,m}(\xi^2;{\bf 1}) \simeq k \xi^{N-m}\oint_\cC\frac{dx}{4\pi i
  x}\frac{(x^{-2};t)_\infty (\xi+\xi^{-1}-x-x^{-1})^N
  x^{-m}}{(a_\xi x^{-1},b_\xi x^{-1},c_\xi x^{-1},d_\xi
  x^{-1};t)_\infty}  
\ee
For $\xi \sim 1$ the function 
$$
f(x)= \log(\xi+\xi^{-1}-x-x^{-1})-\rho_\ast \log(x)
$$  
has a saddle point at
\begin{align}
x_\xi= x_0+ O((\xi-1)),&&x_0=\frac{1+\rho_\ast}{-1+\rho_\ast}\leq -1.
\end{align}
Now recall that for $0<t<1$ the parameters $a,b,c,d$ span the range
$a,c<0$ and  $0<b,d<1$.  When $x_0< a,c< 0$ we can deform the contour of
integration to the steepest descent path that passes through
$x(\rho_\ast;\xi)$ and we easily find 
\begin{multline}\label{asymp1}
Z_{N,m}(\xi^2;{\bf 1}) \simeq
\frac{k  \xi^{N-m}}{\sqrt{8\pi
    f''(x_\xi)N}}\frac{(x_\xi^{-2};t)_\infty
  (\xi+\xi^{-1}-x_\xi-x_\xi^{-1})^N 
  x_\xi^{-m}}{(a_\xi x_\xi^{-1},b_\xi x_\xi^{-1},c_\xi x_\xi^{-1},d_\xi
  x_\xi^{-1};t)_\infty}.
\end{multline}
When at least one among $a$ and $c$ is less then $x_0$ then the
dominating contribution comes from the pole 
around $a_\xi$ if $a<c$ or around $c_\xi$ if $c<a$, and we get
\be\label{asymp2}
Z_{N,m}(\xi^2;{\bf 1}) \simeq \left\{
\begin{array}{cc}
\frac{k  \xi^{N-m}(a_\xi^{-2};t)_\infty
  (\xi+\xi^{-1}-a_\xi-a_\xi^{-1})^N 
  a_\xi^{-m}}{2a_\xi(t,b_\xi a_\xi^{-1},c_\xi a_\xi^{-1},d_\xi
  a_\xi^{-1};t)_\infty}& a<c,x_0\\[17pt] 
 \frac{k  \xi^{N-m}(c_\xi^{-2};t)_\infty
   (\xi+\xi^{-1}-c_\xi-c_\xi^{-1})^N 
  c_\xi^{-m}}{2c_\xi(t,b_\xi c_\xi^{-1},a_\xi c_\xi^{-1},d_\xi
  c_\xi^{-1};t)_\infty}& c<a,x_0
\end{array}
\right.
\ee

Let us now call $J_{\rho_\ast}$ and $\rho^\bullet_{\rho_\ast}$ the
steady state current and density of first class particle in the
thermodynamic limit. Combining eqs.(\ref{av-dens},\ref{av-curr}) with
eqs.(\ref{asymp1},\ref{asymp2}) we obtain 
\begin{align}
J_{\rho_\ast} =\frac{(t^\fr12 -t^{-\fr12})(1-\rho_\ast^2)}{4},  &&
\rho^\bullet_{\rho_\ast}=\frac{1-\rho_\ast}{2},  && \textrm{for}&& x_0<a,c\\
J_{\rho_\ast} =\frac{a(t^{-\fr12}-t^\fr12)}{(1-a)^2},  &&
\rho^\bullet_{\rho_\ast}=\frac{a}{a-1}-\rho_\ast,  && \textrm{for}&& a<x_0,c\\
J_{\rho_\ast} =\frac{c(t^{-\fr12}-t^\fr12)}{(1-c)^2},  &&
\rho^\bullet_{\rho_\ast}= \frac{1}{1-c}, && \textrm{for}&& c<x_0,a
\end{align}
Upon redefinition $q=t^{-1}, a\leftrightarrow -a^{-1}, d\leftrightarrow
-d^{-1},b\leftrightarrow-c^{-1}$ our results  are in agreement with \cite{uchiyama2008two}.

\section{Conclusion}

In this paper we have analyzed the steady state of the Asymmetric
Simple Exclusion process 
with open boundaries and second class particles by deforming it
through the introduction of spectral parameters, in a way which is
dictated by the integrable structure of the model. The (unnormalized)
probabilities of the particle configurations get promoted to Laurent
polynomials in the spectral parameters, that are constructed in
terms of non-symmetric Koornwinder polynomials. Moreover we have shown
that the partition function coincides with a symmetric
Macdonald-Koornwinder polynomial. As an outcome we have computed the
steady current and the average density of first class particles. 
It is interesting to pursue further the analysis of the inhomogeneous
model and compute other quantities like correlation functions.

In a recent preprint \cite{corteel2015macdonald} Corteel and Williams
have uncovered a different connection between open 2ASEP with open
boundaries a second class particles and the theory of Koornwinder polynomials. It
would be extremely interesting to investigate the relation of their
findings with the algebraic structure presented in our paper.

It is clear that our approach can be applied to other stochastic
interacting particle systems  as soon as one recognizes that the
Markov matrix is integrable. 
The most straightforward generalization of
the present paper would consist in a system with more than two species
of particles and exchange rates that arise from a higher rank quotient
of the affine Hecke algebra of type $\tilde \cC_N$. 
A more interesting class of systems that can be treated with our
approach \cite{cantini-open2}, arise from the
classification of the stochastic boundary scattering matrices of
Crampe~et~al. \cite{crampe2014integrable}. 
One could also consider systems with periodic
boundary conditions. While for just a single class of particles the
steady state is simply uniform, this is no longer true if 
one allows for more classes of particles, and the nontrivial steady
state can be analyzed using the approach presented here
\cite{cantini2015matrix,cantini-TASEP}.

\appendix

\section{Koornwinder, Macdonald-Koornwinder and Askey-Wilson polynomials}\label{algebraic}

In this Appendix we shall the definitions of Koornwinder,
Macdonald-Koornwinder and Askey-Wilson polynomials, and a few of their
properties we needed in the paper. 

\subsection{Nonsymmetric Koornwinder polynomials}

In order to introduce the Non-symmetric Koornwinder polynomials one  
considers the action, in the Noumi representation,  of the commutative
sub-algebra $\cY_N$ on $\bC[z_1^{\pm 1},\dots,z_N^{\pm 1}]$. It turns
out that such action is triangular with respect to the order
$\preceq$ on the monomials  
\cite{stokman2000koornwinder}
\be\label{triangular}
Y_i \bz^\alpha  = \omega_i(\alpha) \bz^\alpha +  \sum_{\bz^\beta
  \prec ~\bz^\alpha} \tilde c_\beta \bz^\beta.
\ee
in order to write the eigenvalues $\omega_i(\alpha)$ we need to
introduce $w_\alpha$, which is the minimal length element in
$\cW_N^{\pp 0}$ such that $w_\alpha 
\alpha^+=\alpha$ and $\rho=\{1,2,\dots,N\} \in \bZ$. Then we can write  
$$
\omega_i(\alpha) = q^{\alpha_i}\left(\sqrt{q^{-1}abcd}\pp t^{N-|(w_\alpha
  \rho)_i|}\right)^{\textrm{sign}((w_\alpha
  \rho)_i)}
$$
The  common eigenfunctions of $\cY_N$ are
the \emph{non-symmetric Koornwinder polynomials} $E_\alpha(\bz)$
\cite{sahi1999nonsymmetric}
\be\label{def-nonsymm-mcd}
\begin{split}
E_\alpha(\bz) &= \bz^\alpha + \sum_{\bz^\beta \prec ~\bz^\alpha} c_\beta \bz^\beta,\\
Y_i E_\alpha(\bz) &= \omega_i(\alpha)E_\alpha(\bz).
\end{split}
\ee
\vskip .5cm
\subsection{Macdonald-Koornwinder polynomials}\label{symm-macd}

For a partition $\lambda \in \bZ^N$, call $H_\lambda$ the subspace of
$\bC[z_1^\pm,\dots,z_N^\pm]$ generated by the polynomials
$E_\alpha(\bz)$ such that $\alpha^+=\lambda$ is preserved by the
action of $\cH_N$ and for generic values of the parameters
$a,b,c,d,q,t$ it forms an irreducible representation. The subspace
$H_\lambda^{\cW_N^{\pp 0}}$, invariant under the action of $\cW_N^{\pp 0}$  is
one-dimensional and is generated by $P_\lambda(\bz)$, the
\emph{symmetric Macdonald-Koornwinder polynomial} in $N$ variables
$\bz$, associated to the partition $\lambda$
\cite{koornwinder1992askey, sahi1999nonsymmetric}. 
Before providing a definition and discussion of
Macdonald-Koornwinder polynomials let us state one of the main results we need.
\begin{theorem}[\cite{sahi1999nonsymmetric}, Corollary 6.6]\label{theo-sahi}
Let $\chi:\cH_N\rightarrow \bC$ be the one dimensional representation
of  $\cH_N$ defined on the generators by 
$$
\chi(T_i)=t_i^{\fr12}.
$$
Then we have
\be
P_{\alpha^+}(\bz) \propto \sum_{w\in W_N^{0}} \chi(T_w)T_w E_\alpha(\bz)
\ee
\end{theorem}
For a fixed $N$, the Macdonald-Koornwinder polynomials are Laurent
polynomials in $N$ variables, invariant 
under $\cW_N^{\pp 0}$, which  
are eigenvectors of the following q-difference operator
$$
{\mathcal D}_{q,t}= \sum_{i=1}^N
\Phi_i(z_i)(T_{q,z_i}-1)+\Phi_i(z_i^{-1})(T^{-1}_{q,z_i}-1) 
$$ 
where $T_{q,z_i}$ is the $i$-th q-shift operator
$$
T_{q,z_i}f(z_1,\dots,z_i,\dots,z_N)= f(z_1,\dots,qz_i,\dots,z_N)
$$
and 
$$
\Phi_i(z)=
\frac{(1-az)(1-bz)(1-cz)(1-dz)}{(1-z^2)(1-qz^2)}\prod_{\substack{j=1\\j\neq
    i}}^N\frac{(1-tzz_j)(1-tzz_j^{-1})}{(1-zz_j)(1-zz_j^{-1})}
$$
The Macdonald-Koornwinder polynomials in $N$ variables,
$P_\lambda(\bz)$, are labeled by  partitions $N$ parts, which
parametrize the eigenvalue of ${\mathcal D}_{q,t}$, i.e. the
polynomials $P_\lambda(\bz)$ are characterized by the equations
$$
{\mathcal D}_{q,t}P_\lambda(\bz) = d_\lambda P_\lambda(\bz)
$$
with 
$$
d_\lambda =\sum_{i=1}^N
\left[q^{-1}abcdt^{2n-i-1}(q^{\lambda_i}-1)+t^{i-1}(q^{-\lambda_i}-1)
  \right] 
$$
together with the condition that the coefficient of
$\bz^\lambda$  in $P_\lambda(\bz)$ is $1$. Of course the
Macdonald-Koornwinder polynomials depend on the parameter
$a,b,c,d,t,q$ but we shall write this dependence explicitly
(writing $P_\lambda(\bz)=P_\lambda(\bz;a,b,c,d;q;t)$) only when needed.

An important property of the Macdonald-Koornwinder polynomials (which
could be used actually to give an alternative definition) is their
orthogonality  with respect to a certain scalar product.
Let assume for convenience $|q|<1$ and recall the usual notation for
the q-Pochhammer symbol  
\begin{align}
(a;q)_n&=\prod_{i=0}^{n-1}(1-a
q^i) 
&(a_1,a_2\dots,a_j;q)_n=\prod_{i=1}^j(a_i;q)_n.&
\end{align}
Let us define the kernel
\be
\Delta_N(t,q;\bz) = \prod_{\substack{1\leq i<j\leq
    N\\\epsilon_1,\epsilon_2=\pm1}}\frac{(z_i^{\epsilon_1}z_j^{\epsilon_2};q)_\infty}{(t
  z_i^{\epsilon_1}z_j^{\epsilon_2};q)_\infty} \prod_{i=1}^N
\frac{(z_i^{2\epsilon};q)_\infty}{(az_i^{\epsilon}, bz_i^{\epsilon},
  cz_i^{\epsilon}, dz_i^{\epsilon};q)_\infty},  
\ee
and the scalar product on $\bC[z_1^\pm,\dots,z_N^\pm]$
$$
\langle p,q\rangle:= \oint_\cC \prod_{j=1}^N\frac{dz_j}{4\pi i z_j}
\Delta_N(t,q;\bz) p(\bz)q(\bz^{-1}) 
$$
where the contour of integration $\cC$ encircles the poles in
$aq^k,bq^k,cq^k,dq^k$ ($k\in \bZ_+$) and excludes all the others. 
The Macdonald-Koornwinder polynomials corresponding to different
partitions are orthogonal with respect to such scalar product
\begin{align}
&\langle P_\lambda,P_\mu\rangle =0 &\textrm{if}&&\lambda\neq \mu.
\end{align}

When the partition $\lambda=\{m\}$, i.e. when $\lambda$ consists of
just a single part of length $m$, the
Macdonald-Koornwinder polynomials are polynomials in one variable
$z$ independent of the variable $t$. They correspond (up to a factor and change of
variables $x=(z+z^{-1})/2$) to the Askey-Wilson polynomials
$p_m(z;a,b,c,d|q)$\footnote{In the literature the Askey-Wilson
  polynomials are usually considered as function of $x=\frac{z+z^{-1}}{2}$.}
\cite{askey1985some,gasper2004basic}
\be
P_{\{m\}}(\bz) \propto p_m(z;a,b,c,d|q).
\ee

\subsection{Askey-Wilson polynomials}\label{AW-section}
Let us recall the definition of the Askey-Wilson polynomials in terms of a terminating  
hypergeometric function 
\be
p_n(z;a,b,c,d|q) = \frac{(ab,ac,ad;q)_n}{a^n} {}_4 \phi_3
\left(\begin{array}{c} q^{-n},q^{n-1}abcd,az,az^{-1}\\
    ab,ac,ad\end{array};q,q\right) 
\ee
where the basic hypergeometric functions ${}_r \phi_s $ are defined by 
\be
{}_r \phi_s \left(\begin{array}{c}a_1,a_2,\dots,a_s\\ b_1,\dots,
b_r\end{array};q,z\right) = \sum_{n=0}^\infty
\frac{(a_1,\dots,a_r;q)_n}{(q,b_1,\dots,b_s;q)_n} \left[(-1)^nt^{\binom{n}{2}}\right]^{s+1-r}z^n
\ee
As for the case of Macdonald-Koornwinder
polynomials, whenever the parameters $a,b,c,d$ are clear
from the context we shall use a lighter notation writing
$p_n(z|q)=p_n(z;a,b,c,d|q)$. The Askey-Wilson polynomials satisfy 
orthogonality relations that correspond to the case $N=1$ of the
orthogonality relation for Koornwinder-Macdonald polynomials. 
Indeed for $N=1$ the kernel $\Delta_1(t,q;\bz)$ becomes $t$ independent and
coincides with the Askey-Wilson kernel 
\be
w(z;a,b,c,d|q)=\frac{(z^2,z^{-2};q)_\infty}{
  (az,az^{-1},bz,bz^{-1},cz,cz^{-1},dz,dz^{-1};q)_\infty}   
\ee
and one has
\be
\oint_{\cC} \frac{dz}{4\pi
  iz}w(z|q)p_n(z|q)p_m(z|q) = h_n \delta_{n,m}.
\ee
where 
$$
h_n=\frac{(q^{n-1}abcd;q)_\infty}{(1-q^{2n-1}abcd)(q^{n+1},q^nab,q^nac,q^nad,q^nbc,q^nbd,q^ncd;q)_\infty}.
$$
An important information we need in the paper is the asymptotic form
of $p_m(z|q)$ for large values of the label $m$
\cite[Eq.7.5.13]{gasper2004basic} 
\be\label{asympt-AW}
p_m(z|q) \sim z^m B(z^{-1})+z^{-m}B(z)
\ee
with 
$$
B(z)= \frac{(az,bz,cz,dz;q)_\infty}{(z^2;q)_\infty}.
$$
The specialization $c=d=0$ of the Askey-Wilson polynomials gives the
so called Al-Salam Chihara polynomials
\cite{al1976convolutions,koekoek1996askey}  
\be
Q_n(z;a,b|t) = p_n(z;a,b,0,0|t)
\ee
The Al-Salam Chihara polynomials satisfy the following three terms
recursion for  
\be\label{al-salam-chi-rec}
Q_{n+1}(z)+\left((a+b)t^n-2z\right)Q_{n}(z)+(1-t^n)(1-abt^{n-1})Q_{n-1}(z)=0.
\ee
The recursion relations presented in Section \ref{sect-rec-rel}
motivate the presentation of several contiguous/difference relations for Askey-Wilson
polynomials, which are quite straightforward to prove. They are best stated using divided difference operators
$$
\partial_{x,y}F(x,y):=
\frac{F(x,y)-F(y,x)}{(x-y)},~~~~\partial_{x}F(x):=
\frac{F(x)-F(x^{-1})}{(x-x^{-1})}. 
$$
and using the variable $z$ instead of $x= \frac{z+z^{-1}}{2}$
\begin{align*}
P_m(z;a,b,c,d)&=\partial_{c,d}\left(\omega^{(m,m)}(z;a,b,c,d)
P_m(z;a,b,qc,d)\right)\\
P_{m-1}(z;a,b,qc,qd)&=\partial_{c,d}\left(\omega^{(m-1,m)}(z;a,b,c,d)
P_m(z;a,b,qc,d)\right) \\
P_{m+1}(z;a,b,c,d)
&=\partial_{c,d}\left(\omega^{(m+1,m)}(z;a,b,c,d)
  P_m(z;a,b,qc,d) \right)\\
P_{m}(z;a,b,c,qd) &=
\partial_{z} \left(\theta^{(m,m)}(z;a,b,c,d) P_{m}(q^\fr12
z;q^\fr12 a,q^\fr12 b,q^\fr12 c,q^\fr12 d)\right) \\
P_{m+1}(z;a,b,c,d)&= \partial_{z} \left(\theta^{(m+1,m)}(z;a,b,c,d) P_{m}(q^\fr12
z;q^\fr12 a,q^\fr12 b,q^\fr12 c,q^\fr12 d)\right)
\end{align*}
with
\begin{align*}
\omega^{(m,m)}(z;a,b,c,d)&= \frac{cd(z-c)(zc-1)}{(1-cdq^m)zc}\\
\omega^{(m-1,m)}(z;a,b,c,d)&=\frac{1}{(1-q^m)(1-abq^{m-1})}\\
\omega^{(m+1,m)}(z;a,b,c,d)&=
\frac{(1-adq^m)(1-bdq^m)(z-c)(zc-1)}{z}\\
\theta^{(m,m)}(z;a,b,c,d)&=
\frac{(abcq^m-z)(1-az)(1-bz)(1-cz)}{q^{\frac{m}{2}}(1-abq^m)(1-acq^m)(1-bcq^m)z^2}\\ 
\theta^{(m+1,m)}(z;a,b,c,d)&=-\frac{(1-az)(1-bz)(1-cz)(1-dz)}{q^{\frac{m}{2}}z^2}.
\end{align*}

\subsection{Cauchy formula for Macdonald-Koornwinder polynomials}
We conclude this Appendix with the Cauchy formula for
Macdonald-Koornwinder polynomials, 
due to Mimachi \cite{mimachi2001duality}. Let  
\be
\Pi(\bz,\bx) = \prod_{\substack{1\leq i \leq N\\1\leq j \leq M}}
  (z_i+z_i^{-1}-x_j-x_j^{-1})
\ee 
and for a partition $\lambda\subseteq M^N$, i.e. a partition made of $N$ parts such that
$\lambda_i\leq M$, define $\bar  \lambda$, a partition of $M$ parts
given by
$$
\bar \lambda_j = \#\{i|\lambda_i <N-i \}.
$$
\begin{theorem}[\cite{mimachi2001duality}, Theorem 2.1]
The Macdonald-Koornwinder polynomials $P_\lambda(z)$ satisfy the equality 
\be
\Pi(\bx,\bz) = \sum_{\lambda \subseteq M^N}(-1)^{|\bar \lambda|}
P_\lambda(\bz; q;t)P_{\bar \lambda}(\bx; t;q).
\ee
\end{theorem}
In the present paper we shall need just the case the case $M=1$, which
combined with the orthogonality property, allows to provide an
integral formula for the Macdonald-Koornwinder polynomials
corresponding to a partition with a single column 
$$
\eta(N,m)=\{\underbrace{1,\dots,1}_{N-m},\underbrace{0,\dots,0}_{m} \} 
$$
involving the Askey-Wilson polynomials of basis $t$ 
\be
P_{\eta(N,m)}(\bz; q;t) = \langle \left(x+x^{-1}\right)^m,p_{m}\rangle^{-1} \oint_\cC
\frac{dx}{4\pi i x} \Pi(\bz,x) w(x|t) p_{m}(x|t) 
\ee
where
\be
\langle \left(x+x^{-1}\right)^m,p_{m}\rangle =
\frac{(abcdq^{2m};q)_\infty}{(q^{m+1},abq^m,acq^m,adq^m,bcq^m,bdq^m,cdq^m;q)_\infty}.  
\ee

\bibliographystyle{amsplain}

\bibliography{bibliografia}

\end{document}